\documentclass[letterpaper]{article} 
\usepackage{aaai25}  
\usepackage{times}  
\usepackage{helvet}  
\usepackage{courier}  
\usepackage[hyphens]{url}  
\usepackage{graphicx} 
\usepackage{natbib}  
\usepackage{caption} 
\title{Decentralized Convergence to Equilibrium Prices in Trading Networks}
\author{
    Edwin Lock\textsuperscript{\rm 1},
    Benjamin Patrick Evans\textsuperscript{\rm 2},
    Eleonora Kreacic\textsuperscript{\rm 2},
    Sujay Bhatt\textsuperscript{\rm 3},
    Alec Koppel\textsuperscript{\rm 3},\\
    Sumitra Ganesh\textsuperscript{\rm 3},
    Paul W. Goldberg\textsuperscript{\rm 1}
}
\affiliations{
    \textsuperscript{\rm 1}Department of Computer Science\\
    University of Oxford\\
    Parks Road\\
    Oxford, OX1 3QD, UK\\
    \{edwin.lock, paul.goldberg\}@cs.ox.ac.uk\\
    \textsuperscript{\rm 2}JPMorgan AI Research\\
    London, UK\\
    \textsuperscript{\rm 3}JPMorgan AI Research\\
    New York, USA\\
    \{benjamin.x.evans, eleonora.kreacic, sumitra.ganesh\}@jpmorgan.com,\\
    \{sujay.bhatt, alec.koppel\}@jpmchase.com
}

\usepackage{graphicx}
\usepackage[USenglish]{babel}
\usepackage{parskip}
\usepackage{amsmath, amssymb, amsthm, mathtools}
\usepackage{mathtools}
\usepackage{bm}
\usepackage[capitalize]{cleveref}
\usepackage{enumerate}
\usepackage[ruled]{algorithm}
\usepackage[noend]{algpseudocode}  
\usepackage{caption}
\usepackage{subcaption}
\usepackage{booktabs}
\usepackage{tikz}
\usepackage{tikz-network}

\tikzset{buying signal/.style={font=\small, anchor=south west, pos=0}}
\tikzset{trade label/.style={font=\normalsize, anchor=north, pos=0.5}}
\tikzset{selling signal/.style={very near end, font=\small, anchor=south east, pos=1}}
\tikzset{agent/.style={circle, draw=black, fill=white, inner sep=0pt, outer sep=4pt, minimum size=5pt}}

\DeclareMathOperator{\argmax}{argmax}
\newcommand{\pb}{\bm{p}}
\newcommand{\offers}{\bm{\sigma}}
\newcommand{\offer}{{\sigma}}

\newcommand{\Z}{\mathbb{Z}}

\newcommand{\market}{\mathcal{M}}

\newtheorem{theorem}{Theorem}

\newtheorem{lemma}[theorem]{Lemma}
\newtheorem{proposition}[theorem]{Proposition}
\newtheorem{corollary}[theorem]{Corollary}
\theoremstyle{definition}

\newtheorem{definition}[theorem]{Definition}
\newtheorem{example}[theorem]{Example}
\newtheorem{observation}[theorem]{Observation}
\newtheorem{conjecture}[theorem]{Conjecture}

\usepackage{ifthen}
\newboolean{commentsactivated}
\setboolean{commentsactivated}{true}
\newcommand{\el}[1]{\ifthenelse{\boolean{commentsactivated}}{{\color{blue} {\em EL: #1 }}}{}}

\definecolor{darkgreen}{rgb}{0,0.5,0}
\newcommand{\pwg}[1]{\ifthenelse{\boolean{commentsactivated}}{{\color{darkgreen} {\em PG: #1 }}}{}}

\newcommand{\preprint}{}

\ifx\preprint\undefined
\else
\nocopyright
\fi

\begin{document}

\maketitle

\begin{abstract}
We propose a decentralized market model in which agents can negotiate bilateral contracts. This builds on a similar, but centralized, model of trading networks introduced by Hatfield et al.~in 2013. Prior work has established that fully-substitutable preferences guarantee the existence of competitive equilibria which can be centrally computed. Our motivation comes from the fact that prices in markets such as over-the-counter markets and used car markets arise from \textit{decentralized} negotiation among agents, which has left open an important question as to whether equilibrium prices can emerge from agent-to-agent bilateral negotiations. We design a best response dynamic intended to capture such negotiations between market participants. We assume fully substitutable preferences for market participants. In this setting, we provide proofs of convergence for sparse markets (covering many real world markets of interest), and experimental results for more general cases, demonstrating that prices indeed reach equilibrium, quickly, via bilateral negotiations.  Our best response dynamic, and its convergence behavior, forms an important first step in understanding how decentralized markets reach, and retain, equilibrium.
\end{abstract}

\section{Introduction}

Price discovery is a central feature of modern markets. In many real-world markets, prices arise from a decentralized process governed by negotiations between market participants. By contrast, market outcomes are classically studied from the perspective of `static' optimality criteria such as competitive equilibrium or stability, or viewed through the lens of centralized market procedures (such as auctions and clearing houses). In particular, this does not address the question of when, and how, desirable market outcomes can arise from \textit{decentralized} market processes.

This paper introduces a `best response' market dynamic that captures uncoordinated negotiations between neighbors in a network of market participants. We work in a general model of trading introduced by \citet{hatfield2013stability}, in which a network of heterogeneous agents can engage in bilateral trades for indivisible goods or services. Agents can act as buyer, sellers, or as traders involved in some trades as a buyer and in other trades as a seller.   Our model allows for cycles of trades, and for `horizontal' as well as `vertical' \citep{ostrovsky2008stability} network structures. The model also embeds well-known settings such as labor markets \citep{kelso1982job}, and general exchange economies with heterogeneous indivisible goods \citep{gul1999walrasian,yang2000equilibrium}, capturing a broad range of market configurations.

Like \citeauthor{hatfield2013stability}, we focus primarily on agents with fully-substitutable (FS) preferences. Our central point of departure from \citet{hatfield2013stability} is that each agent maintains an offer for each of their trades, so each trade has a buying \textit{and} selling offer. In our market dynamic, agents 
take turns to update their offers in a manner that maximizes their utility. The dynamic ends in equilibrium if there is no agent who wishes to modify her offers any further (because her current offers already maximize her utility); at this point the buyer and seller of each trade agree to execute this trade iff their two offers coincide.

Trading networks have received increased attention in recent years \citep{morstyn2018bilateral,https://doi.org/10.3982/TE3240,osogami2023learning}. Prior work has focused on establishing the existence of competitive equilibrium prices and stable outcomes in markets of increasing generality under the condition that agents have substitutes preferences (see, e.g., \citealt{kelso1982job,gul1999walrasian,jagadeesan2021matching,baldwin2023equilibrium}). In the model of \citet{hatfield2013stability}, the authors show that competitive equilibra are guaranteed to exist if agents have {fully-substitutable} preferences. Moreover, an equilibrium is efficiently computable via an ascending-price algorithm. This algorithm can be run, e.g., by a central exchange to centrally compute prices if it has access to all agents' preferences.

But in many real world markets, such as over-the-counter financial markets and peer-to-peer markets, prices are not computed centrally, and preferences are not public. Instead, prices arise from uncoordinated bilateral negotiations of agents with their neighbors in a trading network. Such decentralized processes have been studied in the context of specialized two-sided markets; e.g., the labor markets \citep{chen2016decentralized,chen2016random} and two-sided markets of buyers and sellers with unit demand for identical items \citep{assadi2017fast}. While the decentralized market dynamics proposed in these works converge to competitive equilibrium, they are not immediately generalizable to markets on more general networks. Moreover, to the best of our knowledge, no similar results are known for market topologies beyond two-sided markets. 

This leaves open the important research questions of formulating a decentralized market dynamic that captures the notion of negotiations between neighboring agents in general trading networks, and understanding the dynamic's \mbox{(non-)}convergence to an equilibrium.

\paragraph{Main Contributions.}
Extensive experiments suggest that our proposed market dynamic converges to equilibrium for any \emph{market with fully substitutable agents (FSM)}. \Cref{section:theory} develops theoretical results to support this observation. 
Our main technical contribution is a reduction from markets with arbitrarily many agents to 2-agent markets (\cref{proposition:convergence-guarantees}). This reduction guarantees that FSMs with $m$-sparse networks converge (almost surely) to equilibrium iff FSMs with two agents and $m$ trades converge (almost surely).%
\footnote{A network is \textit{$m$-sparse} if every induced subgraph of the network can be divided into two or more disjoint components by removing at most $m$ edges. See \cref{definition:sparse} for details.}

As we show that 2-agent FSMs with one or two trades converge (\cref{proposition:two-agents-convergence}), our reduction thus implies that markets with tree topologies and 2-sparse FSMs reach equilibrium under our market dynamic. We conjecture, based on our experimental evidence, that $m$-sparse FSMs converge for any $m > 2$.
Surprisingly, we see that full substitutability is not required to guarantee convergence for markets with tree topologies; agents can have arbitrary quasilinear preferences (\cref{theorem:market-convergence}). By contrast, we give a counterexample of a two-agent market with a non-fully-substitutable agent for which the dynamic does not terminate. 

\Cref{section:experiments} complements our theoretical results with multi-agent simulations to develop a better quantitative understanding of convergence in our market dynamic.%
\footnote{The code for our simulations is available upon request from the authors.}
Specifically, we demonstrate the convergence, and speed of convergence, in markets with tree topologies, and move beyond the theoretical results to also demonstrate convergence in general (graph-based) markets. We explore numerous market configurations, including various market sizes and agent compositions, to analyze convergence and equilibrium properties. Our experiments demonstrate that equilibrium is reached {much faster than the identified bounds}. The experiments also highlight natural and desirable properties of our market dynamic, such as a general increase in welfare and satisfaction as we reach equilibrium, {and the increase in utility as competition reduces (and vice versa)}. We also show {how shocks propagate through the market, with knock-on effects before converging to a new equilibria, with the convergence speed governed by the severity of the shock.}

The remainder of the paper is presented as follows. \cref{section:model} proposes the model, \cref{section:theory} provides theoretical proofs, and \cref{section:experiments} analyzes the simulation results. Conclusions and next steps are presented in \cref{section:conclusion}.

\section{Model}
\label{section:model}
A market $\market = (I, \Omega, v)$ consists of a finite set $I$ of agents, a set $\Omega$ of possible bilateral trades and agent valuations $v = (v^i)_{i \in I}$. The trades $\Omega$ can represent any goods or services (such as the sale of a batch of coffee beans, an insurance contract, or a spectrum license), and are typically heterogeneous. Each trade $\omega$ has a single buyer $b(\omega) \in I$ and a single, distinct, seller $s(\omega) \in I$. Subsets of $\Omega$ are \textit{bundles}. $\Omega$ can contain multiple trades with the same buyer and seller, and an agent can be a buyer for some trades and a seller for others. So $\Omega$ defines a directed multi-graph on vertices $I$ in which every trade is an arc. \Cref{example:coffee-network} illustrates our model.

\begin{example}
\label{example:coffee-network}
The coffee supply chain consists of coffee farms, coffee roasters, coffee shops and supermarkets, forming the different types of agents. Roasters source raw beans from farms and process them to various varieties of blends, whilst coffee shops and supermarkets can source beans both from the roasters (processed blends) or farms (raw), representing the goods. A coffee shop/supermarket reaches out to farms and roasters in its network in order to obtain their asking price, and subsequently provides a set of buying offers (the trades). \Cref{fig:coffee} shows a concrete example. Notice the heterogeneous nature of goods the roaster trades: it acts as a buyer of raw product and a seller of various blends, as well as the heterogeneity of the supermarket and coffee shop which have different supplier networks and (private) preferences.
\end{example}

For any agent $i \in I$ and bundle $\Psi \subseteq \Omega$ of trades, $\Psi_i$ denotes the trades in $\Psi$ that involve agent $i$; and $\Psi_{i \rightarrow}$ and $\Psi_{i \leftarrow}$ respectively denote the agent's `selling' and `buying' trades in $\Psi_i$. For each agent $i$ and trade $\omega$, let $\chi^i_\omega=1$ if $i$ is the buyer of $\omega$, $\chi^i_{\omega}=-1$ if $i$ is the seller, and $\chi^i_{\omega}=0$ otherwise.

\begin{figure}[!tb]
\includegraphics[width=.95\columnwidth]{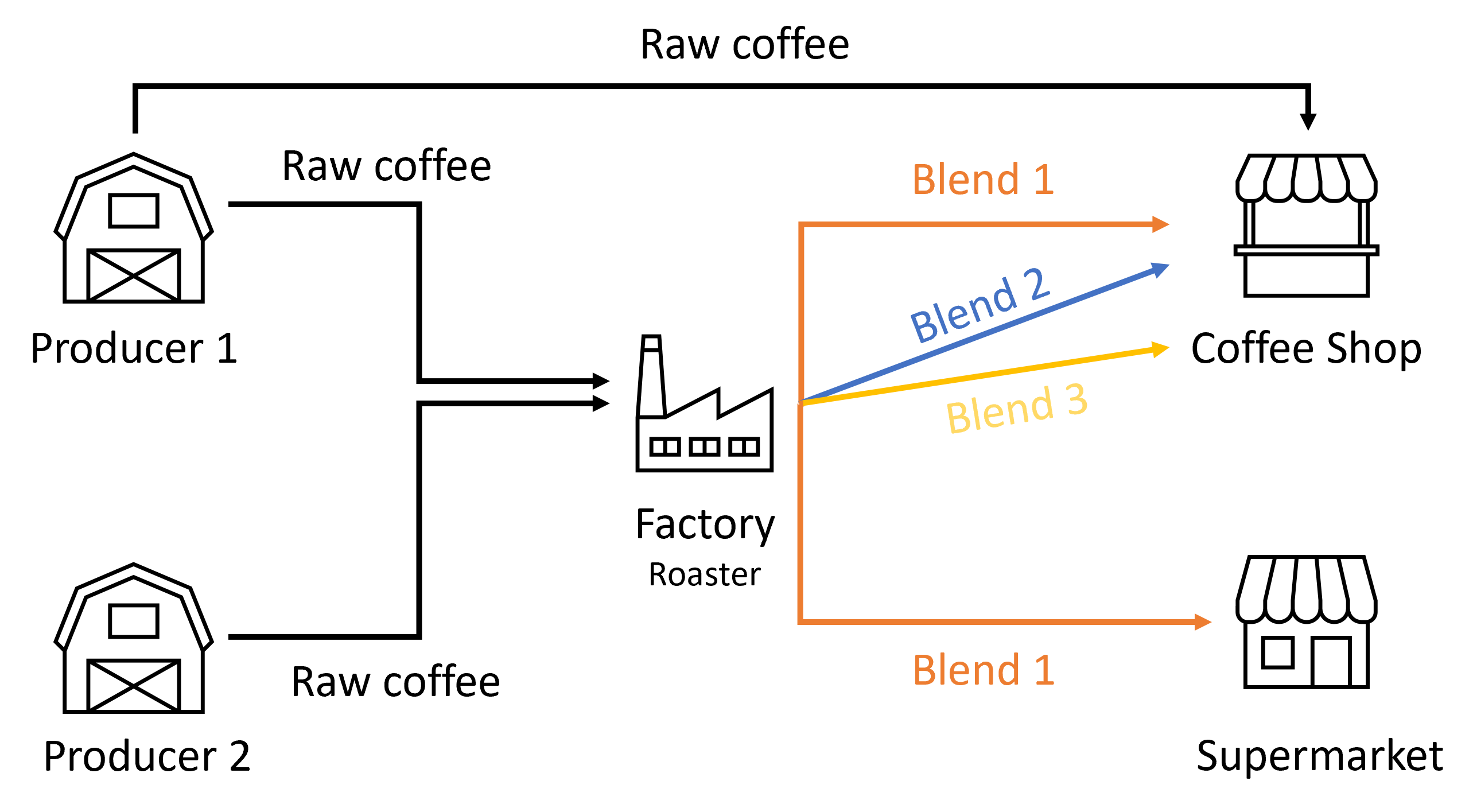}
\caption{An illustration of a market with two producers (coffee farms), one intermediary (a roaster), and two buyers (coffee shop and supermarket).}
\label{fig:coffee}
\end{figure}

Each agent $i$ maintains an integer \textit{offer} $\offer^i_\omega$ for each of her trades~$\omega \in \Omega_i$. The offer $\offer^{b(\omega)}_\omega$ is an offer by the buyer $b(\omega)$ to \textit{buy}, while $\offer^{s(\omega)}_\omega$ is an offer by the seller $s(\omega)$ to \textit{sell}. For any agent $i$ and trade $\omega \in \Omega_i$, we write the offer of the other agent---its \textit{counterpart}---as $\offer^{-i}_{\omega}$. Each agent interprets her counterparts' offers as \textit{prices} $\pb \in \Z^{\Omega_i}$ for trades $\Omega_i$. 

\subsection{Preferences}
\label{section:preferences}
As in \citet{hatfield2013stability}, agents have heterogeneous and private preferences. An agent's private \textit{valuation function} $v^i$ maps every possible bundle $\Psi \subseteq \Omega_i$ of her trades to an integer value,
and the empty bundle $\emptyset$ to $0$.%
\footnote{Values for certain bundles can be negative, e.g.~when they represent production costs for a seller or intermediary agent, as illustrated in \cref{example:coffee-network}.}
The utility $u^i(\Psi, \pb)$ an agent has for bundle $\Psi$ at prices $\pb \in \Z^{\Omega_i}$ is the agent's value for $\Psi$ plus the income from selling trades $\Psi_{i \rightarrow }$ minus the spending on buying trades $\Psi_{i \leftarrow}$:%
\footnote{Heterogeneous preferences and quasilinear utilities are standard in the literature on substitutes markets and auctions \citep{hatfield2013stability,kelso1982job,Aus06}.}
\begin{equation}
\label{eq:utility]}
u^i(\Psi, \pb) \coloneqq v^i(\Psi) - \sum_{\omega \in \Psi} \chi^i_\omega p_\omega.
\end{equation}
This gives rise to an agent's demand $D^i$ consisting of the bundle $\Psi \subseteq \Omega_i$ that maximizes her utility at prices $\pb$:%
\footnote{If an agent $i$ is indifferent between multiple bundles at $\pb$, we assume a consistent tie-breaking rule so that her demanded bundle is unique.}
\begin{equation}
\label{eq:demand}
D^i(\pb) \coloneqq \argmax_{\Psi \subseteq \Omega_i} u^i(\Psi, \pb).
\end{equation}
If an agent $i$ is indifferent between multiple bundles at $\pb$, we assume a consistent tie-breaking rule so that her demanded bundle is unique. In our experiments, we do this by perturbing $v^i$ to achieve unique demand at any integer prices.

We focus on agents with fully-substitutable demands. An agent's demand is \textit{fully-substitutable} if reducing the prices of some buying trades (weakly) decreases her demand for the remaining buying trades and (weakly) increases her demand of selling trades; and, conversely, raising the prices of some selling trades (weakly) decreases her demand for the remaining selling trades and (weakly) increases her demand of buying trades. For agents with only buying trades, or only selling trades, this definition coincides with the standard definition of substitutes introduced by \citet{kelso1982job} in the context of labor markets, also prevalent in the auction literature \citep{Aus06, baldwin2023solving}.

\begin{definition}
\label{definition:full-substitutability}
The demand of agent $i$ is \textit{fully-substitutable} if the uniquely demanded bundles $\Psi, \Psi'$ at any prices $\pb, \pb' \in \Z^{\Omega_i}$ satisfy:
\begin{enumerate}[(i)]
\item $\Psi_{i \rightarrow} \subseteq \Psi'_{i \rightarrow}$ and $\{\omega \in \Psi'_{i \leftarrow} \mid p_\omega = p'_\omega \} \subseteq \Psi_{i \leftarrow}$ when $p_\omega = p'_\omega$ for $\omega \in \Omega_{i \rightarrow}$ and $p_\omega \geq p'_\omega$ for $\omega \in \Omega_{i \leftarrow}$;
\item $\Psi_{i \leftarrow} \subseteq \Psi'_{i \leftarrow}$ and $\{\omega \in \Psi'_{i \rightarrow} \mid p_\omega = p'_\omega \} \subseteq \Psi_{i \rightarrow}$ when $p_\omega = p'_\omega$ for $\omega \in \Omega_{i \leftarrow}$ and $p_\omega \leq p'_\omega$ for $\omega \in \Omega_{i \rightarrow}$.
\end{enumerate}
\end{definition}

We give examples of fully-substitutable buyers, sellers, and traders in the context of \cref{example:coffee-network}. \citet{hatfield2015full,https://doi.org/10.3982/TE3240} provide further equivalent definitions of FS. We call a market with only fully-substitutable agents an \textit{FSM}. 

\begin{example}
\label{example:preferences}
The roaster from \cref{fig:coffee} interested in sourcing raw beans from any one of the two producers, but not both, can express her valuation $v$ as $v(\emptyset) = 0$, $v(\{\omega_{1}\})  = v(\{\omega_{2}\}) > 0$ and $v(\{\omega_{1}, \omega_{2}\}) = -M$, where $\omega_{1}$ and $\omega_{2}$ denote the buying trades of the roaster and $M$ is a sufficiently large negative number. Moreover, the roaster may have the technological constraint that she can only participate in a selling trade (e.g., ground coffee) if she also participates in the buying trade (e.g., coffee beans). She can implement this by assigning value $-M$ to any technologically infeasible bundle of trades, e.g., $v(\{\rho\}) =-M$, where $\rho$ denotes any trade where the roaster acts as a seller.

From the perspective of the coffee shop, a reduction in the price of any blend would reduce its demand for the remaining blends (cf.~\cref{definition:full-substitutability}).
\end{example}

\subsection{Market Dynamic}
\label{section:market-dynamic}
In our market dynamic, agents negotiate with their neighbors by modifying their own offers.
Initially, the offers for trades are arbitrary, and all agents are marked as \textit{unsatisfied}. The dynamic then selects one agent $i \in I$ at a time, uniformly at random, to update her offers by a discrete step size, $\varepsilon \in \mathbb{N}$.\footnote{Without loss of generality, we will assume $\varepsilon = 1$ throughout.}

The agent $i$ first observes her counterparts' offers for its trades $\Omega_i$, which she interprets as the current prices $\pb \in \Z^{\Omega_i}$ of these trades. She then determines the bundle $\Psi \in D^i(\pb)$ she demands at these prices, so that $\Psi$ maximizes her utility $u^i(\cdot, \pb)$. She then \textit{best responds} by updating her offers as follows. As the agent agrees to the prices of all trades $\Psi$, she sets her offers $\offer^i_\omega$ for $\omega \in \Psi$ to match her counterparts' offers $\offer^{-i}_{\omega}$. For all trades $\omega \in \Omega_i \setminus \Psi$, which she does not demand, the agent sets her offer $\omega^i_{\omega}$ to be $\varepsilon$ lower than $\offer^{-i}_{\omega}$ if $\omega$ is a buying trade, and $\varepsilon$ higher than $\offer^{-i}_{\omega}$ if $\omega$ is a selling trade. Setting her offers close to the trading partners' offers leaves room for negotiations while indicating her demand.
Once an agent has best responded, she is marked as \textit{satisfied} and all agents faced with a modified offer from agent $i$ are marked as \textit{unsatisfied}.

\begin{figure}[!tb]
 \centering
 \includegraphics[width=.9\columnwidth]{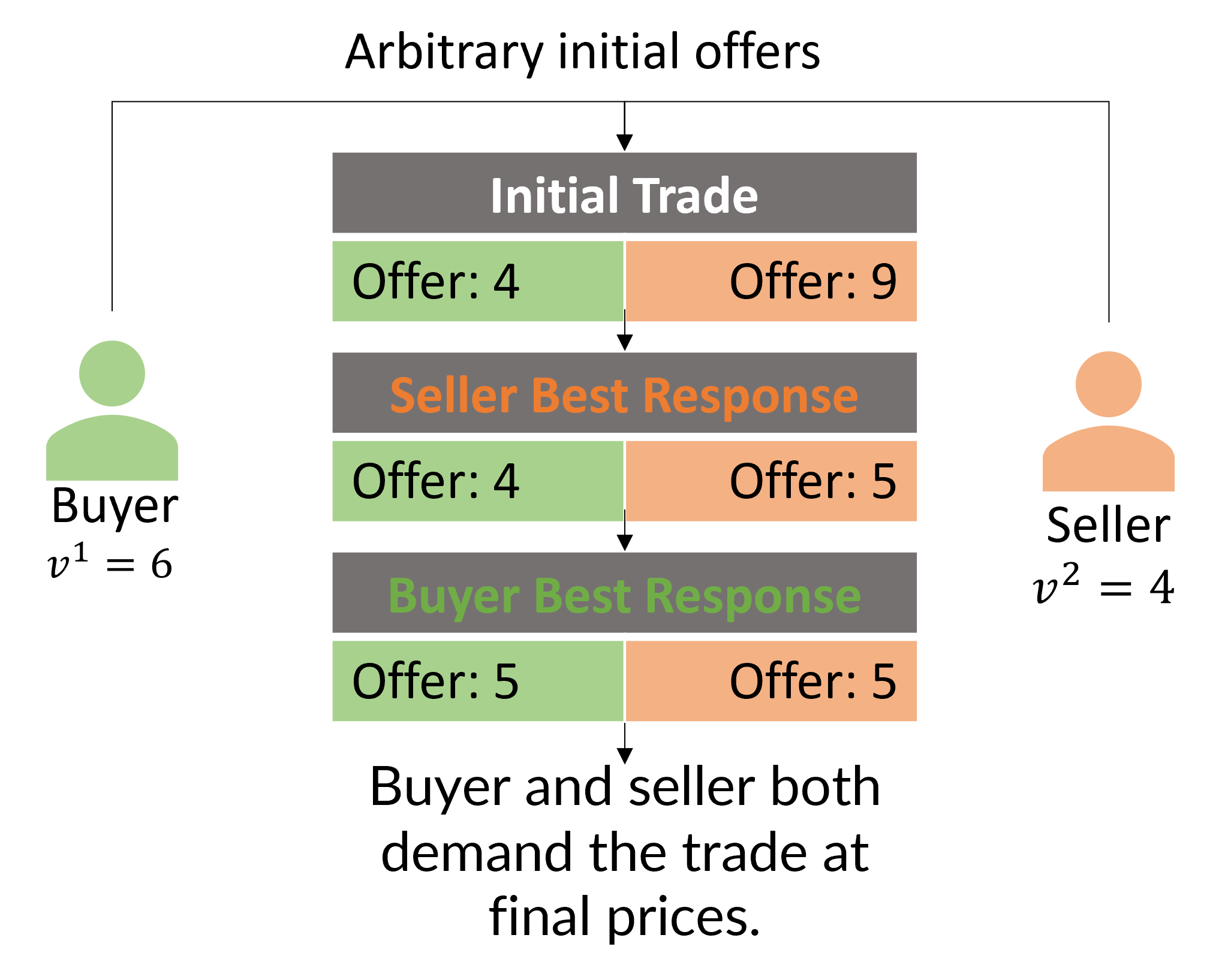}
\caption{An example trade between two parties. Rather than trades having a single price, each trade has two offers, the buyer's offer and the seller's offer. Offers are initially arbitrary, and iteratively updated through best responses.}
\label{figTrade}
\end{figure}

\begin{definition}
\label{definition:best-response}
Suppose an agent $i$ demands $\Phi \in D^i(\pb)$ at the prices $p_\omega = \offer^{-i}_{\omega}$ corresponding to her counterpart's offer for each $\omega \in \Omega_i$. Her \textit{best response} (BR) sets her new offers to
\begin{equation}
\offer^i_{\omega} \coloneqq
\begin{cases}
    \offer^{-i}_{\omega} & \text{if } \omega \in \Phi, \\
    \offer^{-i}_{\omega} - \chi^i_{\omega} \varepsilon & \text{else.}
\end{cases}
\end{equation}
\end{definition}

The \textit{state} $(U,\offers)$ of the market dynamic after each BR is given by the set $U$ of unsatisfied agents and the current offers $\offers$. The dynamic terminates once all agents are satisfied, in which case we say that the market has reached \textit{equilibrium}. \Cref{alg:BR-dynamic} describes the market dynamic in detail. At equilibrium, both agents $i,j$ of each trade $\omega$ demand $\omega$ if $\offer^i_{\omega} = \offer^j_{\omega}$, and both don't demand $\omega$ if $\offer^i_{\omega} \neq \offer^j_{\omega}$.%

As soon as every agent in the market has been selected at least once, it is straightforward that the two offers of any trade differ by $0$ or $\varepsilon$, and in particular satisfy $\offer^{b(\omega)}_\omega \leq \offer^{s(\omega)}_\omega \leq \offer^{b(\omega)}_\omega + \varepsilon$. We restrict our attention to this ``main'' phase of the market dynamic in our theoretical analysis.

\begin{algorithm}[!tb]
\caption{Market dynamic.}
\label{alg:BR-dynamic}
\begin{algorithmic}[1]
\State Initialize buying and selling offers $\offers$ for all trades.
\State Let $U=I$ be the set of unsatisfied agents.
\While {$U \neq \emptyset$}
    \State Sample an agent $i \sim U$ uniformly at random.
    \State Let $\pb \in \Z^{\Omega_i}$ with $p_{\omega} = \offer^{-i}_{\omega}$ for all $\omega \in \Omega_i$ be the offers facing agent $i$.
    \State \label{algstate:smallest-bundle} Determine the unique bundle $\Psi \in D^i(\pb)$ demanded by agent $i$ at prices~$\pb$  [cf.~\eqref{eq:demand}].
    \State Update agent $i$'s offers to $\offer^i_{\omega} = p_\omega$ for $\omega \in \Psi$ and $\offer^i_{\omega} = p_\omega - \chi^i_{\omega} \varepsilon$ for $\omega \in \Omega_i \setminus \Psi$.
    \State Remove agent $i$ from $U$.
    \State Add all agents $j \neq i$ facing modified offers to $U$.
\EndWhile
\State \Return final offers $\offers$.
\end{algorithmic}
\end{algorithm}

We note that fully-substitutable utility functions $u^i(\cdot, \pb)$ are $M^\natural$-concave for any fixed prices $\pb$ (cf.~\citet[Section 4.5]{hatfield2015full}). So agents can compute their demand at any prices $\pb$ in polynomial time \citep[Chapter 10]{murota2003discrete} to determine their best response.

\section{The Emergence of Equilibria}
\label{section:theory}
In extensive numerical experiments, we observe that FSMs always converge to equilibrium. We now present theoretical results towards a proof that FSMs converge almost surely.%
\footnote{Recall that the randomness in the market arises from the choice of unsatisfied agent in each round of the market dynamic.}
This qualitative result is complemented by our experiments in \cref{section:experiments}, which focus primarily on showing the convergence speed and additional welfare analysis. All formal statements and proofs can be found in the appendix. Let $V$ be an upper bound on the absolute value of the agents' valuation functions and initial offers.

Any best response sequence for a market $(I, \Omega, v)$ in a given state is fully determined by the sequence  $i_1, i_2, \ldots$ of agents $i_k \in I$ who are best-responding. In markets with two agents, the best response sequence for any given market and state is unique up to the starting agent, as the agents best respond in alternation. If the market consists of a single trade, then we see that it converges after at most $O(V)$ best responses. For two-agent markets with two trades, we show that full substitutability for both agents guarantees convergence to equilibrium after at most $O(V^2)$ best responses. We note that the definition of FS (\cref{definition:full-substitutability}) only applies in the presence of two or more trades.

\begin{proposition}
\label{proposition:two-agents-convergence}
Markets consisting of a single trade converge to equilibrium after $O(V)$ BRs.
Two-agent FSMs with two trades converge to equilibrium after $O(V^2)$ BRs.
\end{proposition}
We prove the single-trade case by contradiction: suppose that the dynamic does not terminate. Then we observe that one agent always accepts her counterpart's offer for the trade $\omega$, while the other always rejects. So the two agent's offers both increase, or both decrease, by $\epsilon$ after each round of best responses from both agents. But the buyer (seller) demands $\emptyset$ instead of the trade if the price is high (low) enough, a contradiction.
In the two-trade case, we generalize this approach and use full substitutability to reach the required contradiction. We refer to the appendix for the full proof of \cref{proposition:two-agents-convergence}.

By contrast, \cref{example:complements-cycle} presents a market with two trades between a substitutes buyer and a complements seller for which the market dynamic fails to terminate. This demonstrates that FS is a necessary condition in the maximal domain sense, as the market dynamic can cycle if one or more agents are not fully-substitutable. We conjecture that two-agent FSMs converge for any number $m$ of trades after at most $f(m,V) = O(V^m)$ best responses, but that this will require new proof techniques.

\begin{example}
\label{example:complements-cycle}
Consider a two-agent market with a buyer and a seller and valuations given by:
\begin{center}
\begin{tabular}{ccccc}
            & $\emptyset$ & $\{\omega\}$ & $\{\varphi \}$ & $\{\omega, \varphi \}$ \\
            \midrule
    buyer   & $0$ & $8$ & $9$ & $-\infty$ \\
    seller   & $0$ & $-6$ & $-7$ & $-9$
\end{tabular}
\end{center}
The buyer has substitutes preferences, and the seller has complementary preferences that violate \cref{definition:full-substitutability}. If the buyer makes initial offers $\offer^b_\omega=4$ and $\offer^b_\omega = 5$ for the two trades, and the market dynamic starts with a best response from the seller, then the dynamic cycles after two best responses from each agent. Assuming agents apply the lexicographic tie-breaking rule $\emptyset \prec \{\omega \} \prec \{\varphi\} \prec \{\omega, \varphi\}$ to decide their demanded bundle between two or more utility-maximizing bundles, the offers made by the alternating buyer and seller are:
\begin{center}
\begin{tabular}{cccccc}
            &  $\offers^b$ & $\offers^s$ & $\offers^b$ & $\offers^s$ & $\offers^b$ \\
\midrule
$\omega$    &  $4$ & $5$ & $5$ & $5$ & $4$ \\
$\varphi$   &  $5$ & $6$ & $5$ & $5$ & $5$
\end{tabular}
\end{center}
\end{example}

\begin{conjecture}
\label{conjecture:two-agent-convergence}
Two-agent FSMs with any number of trades converge to equilibrium.
\end{conjecture}

We next turn to markets with arbitrarily many agents. Our main theoretical result is to reduce the problem of convergence for markets with any number of agents to the case of two-agent markets. For this, we categorize markets by the `sparsity' of their topologies.

\begin{definition}
\label{definition:sparse}
A market is \textit{$m$-sparse} if every subgraph of its underlying graph admits a cut of size at most~$m$. (A \textit{cut of size $m$} is a partition of a graph's vertices into two vertex sets with $m$ crossing edges between the two sets.)
\end{definition}
The topology of a $1$-sparse market is a forest. Moreover, markets with agents who are each involved in at most $m$ trades are trivially $m$-sparse.

In \cref{proposition:convergence-guarantees}, we show that the convergence of two-agent markets with one trade implies that any $1$-sparse market admits a best response sequence after which the market terminates. Surprisingly, this result holds even if the agents are not fully-substitutable. For $m$-sparse markets with $m \geq 2$, \cref{example:complements-cycle} motivates us to focus on the setting with full substitutability. We establish that every FSM admits a best response sequence after which the dynamic terminates if and only if two-agent FSMs converge.

\begin{figure*}[htb!]
\centering
\begin{subfigure}[b]{0.33\textwidth}
\centering
\begin{tikzpicture}[xscale=0.3, yscale=0.45]
\draw[blue, dashed] (0,0) rectangle (8,6);
\draw[red, dashed] (11,0) rectangle (19,6);
\node[blue] at (4,7) {$J^1$};
\node[red] at (15,7) {$J^2$};
\node[agent, blue] (x1) at (3,5) {};
\node[agent, blue] (x2) at (2,2) {};
\node[agent, blue] (a) at (6,5) {};
\node[agent, blue] (b) at (5,3) {};
\node[agent, blue] (c) at (7,1) {};
\node[agent, red] (d) at (13,5) {};
\node[agent, red] (e) at (14,3) {};
\node[agent, red] (f) at (13,2) {};
\node[agent, red] (y1) at (16,1) {};
\node[agent, red] (y2) at (17,4) {};
\draw[-Latex] (a) -- (d) node[midway,fill=white] {$\omega$};
\draw[Latex-] (b) -- (e) node[midway,fill=white] {$\varphi$};
\draw[-Latex] (c) -- (f) node[midway,fill=white] {$\psi$};
\draw[-Latex] (x1) -- (x2);
\draw[-Latex] (x1) -- (a);
\draw[-Latex] (x2) -- (b);
\draw[-Latex] (c) -- (b);

\draw[-Latex] (d) -- (y2);
\draw[-Latex] (y2) -- (e);
\draw[-Latex] (y1) -- (y2);
\path (f) edge[-Latex, bend left] (y1);
\path (y1) edge[-Latex, bend left] (f);
\end{tikzpicture}
\caption{A $3$-sparse market}
\end{subfigure}
\begin{subfigure}[b]{0.3\textwidth}
\centering
\begin{tikzpicture}[xscale=0.3, yscale=0.45]
\draw[blue, dashed] (0,0) rectangle (8,6);
\node[blue] at (4,7) {$J^1$};
\node[agent,blue] (x1) at (3,5) {};
\node[agent,blue] (x2) at (2,2) {};
\node[agent,blue] (a) at (6,5) {};
\node[agent,blue] (b) at (5,3) {};
\node[agent,blue] (c) at (7,1) {};
\draw[-Latex, dashed, draw=gray] (a) -- (d) node[midway,fill=white] {$\omega$};
\draw[Latex-, dashed, draw=gray] (b) -- (e) node[midway,fill=white] {$\varphi$};
\draw[-Latex, dashed, draw=gray] (c) -- (f) node[midway,fill=white] {$\psi$};
\draw[-Latex] (x1) -- (x2);
\draw[-Latex] (x1) -- (a);
\draw[-Latex] (x2) -- (b);
\draw[-Latex] (c) -- (b);
\end{tikzpicture}
\caption{The market restricted to $J^1$}
\end{subfigure}
\begin{subfigure}[b]{0.3\textwidth}
\centering
\begin{tikzpicture}[xscale=0.3, yscale=0.45]
\draw[draw=none, fill=none] (-2,0) rectangle (12,7);
\node[agent, blue, label=below left:$1$] (one) at (0,3) {};
\node[agent, red, label=below right:$2$] (two) at (10,3) {};
\path[-Latex] (one) edge [bend left] node [midway,fill=white] {$\omega$} (two);
\path[Latex-] (one) edge [] node [midway,fill=white] {$\varphi$} (two);
\path[-Latex] (one) edge [bend right] node [midway,fill=white] {$\psi$} (two);
\end{tikzpicture}
\caption{The market after merging $J^1$ and $J^2$.}
\end{subfigure}
\caption{Illustrating the proof of \cref{proposition:convergence-guarantees}.}
\label{fig:main-proof-illustration}
\end{figure*}
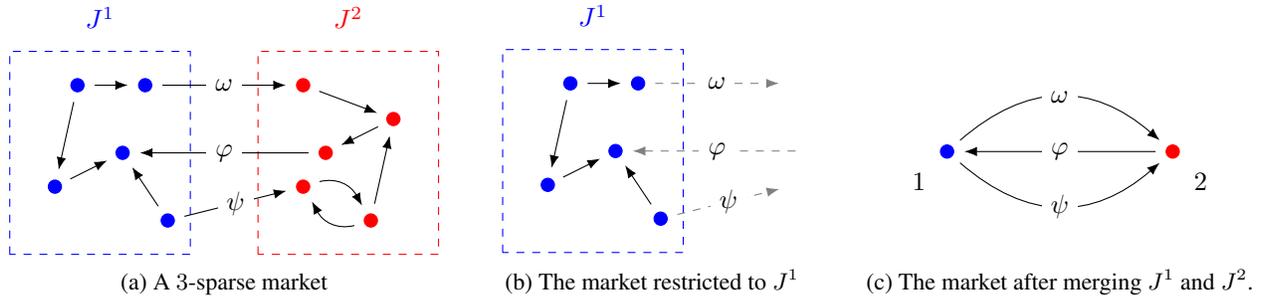

\begin{proposition}
\label{proposition:convergence-guarantees}
\ 
\begin{enumerate}[(i)]
\item Every market with a forest topology (even when agents are not FS), in any state, admits a terminating best response sequence iff two-agent markets with one trade converge.
\item Every $m$-sparse FSM, in any state, admits a terminating best response sequence iff two-agent FSMs with $m$ trades converge.
\end{enumerate}
\end{proposition}

We prove \cref{proposition:convergence-guarantees} in two steps. Fix a market $\market = (I, \Omega, v)$ in state $(U^0, \offers^0)$. We partition the set of agents $I$ of the market into two non-empty subsets, $J^1$ and $J^2$, and let $\Omega^{1,2}$ be the trades between these agent sets. This is illustrated in \cref{fig:main-proof-illustration}, in which $\Omega^{1,2} = \{\omega, \varphi, \psi\}$. For notational convenience, define $J^k \coloneqq J^{k \bmod 2}$ for $k \geq 1$.

In the first step, we describe a procedure to construct a best response sequence $S = S^1 S^2 \ldots$ consisting of subsequences $S^k$. Each $S^k$ contains only agents from $J^k$, and \textit{all} agents in $J^{k}$ are satisfied after applying sequence $S^1 \ldots S^k$. Our second step is then to prove that $S$ is finite and the market terminates after applying $S$. We now sketch out each step in more detail; the full proof is given in the appendix.

The subsequences $S^k$ are constructed recursively. Let $(U^{k-1},\offers^{k-1})$ be the state of the market after applying sequence $S^1 \ldots S^{k-1}$. If all agents in the market are satisfied after this sequence, we are done and define $S = S^1 \ldots, S^{k-1}$. Otherwise, to find $S^{k}$, we restrict the market $\market$ in state $(U^{k-1}, \offers^{k-1})$ to agents $J^k$ and modify the valuations of the agents in $J^k$ who share a trade with agents outside $J^k$. The modified valuations endow agents with the opportunity to execute their trades in $\Omega^{1,2}$ at prices set to their counterparts' offers, even though these are no longer included in the restricted market. This valuation transformation was first described in \citet{hatfield2015full}, and preserves full substitutability. Intuitively, this means that the agents $J^k$ best respond identically in the original and restricted markets. By induction, the restricted market admits a terminating best response sequence $S^k$. As all agents in the restricted market are satisfied after $S^k$, all agents in $J^k$ are also satisfied after applying $S^k$ to the original market in state $(U^{k-1}, \offers^{k-1})$. \Cref{fig:main-proof-illustration} (b) shows the restricted market in which the three agents of $J^k$ involved in the grayed-out trades $\Omega^{1,2}$ have valuations that endogenise these trades.

In order to show that $S$ is finite, we consider the market $\widetilde{\market}$ obtained by merging agent sets $J^1$ and $J^2$ into two agents, $1$ and $2$. The trades of $\widetilde{\market}$ consist of $\Omega^{1,2}$ with endpoints changed in the natural way to agents $1$ and $2$, and the initial market state is modified similarly. Moreover, the valuation of merged agent $k$ is defined such that she demands the same trades $\Omega^{1,2}$ as the agents $J^k$ in the original market do, whenever the counterparts' offers for these trades are the same in both markets, preserving FS \citep{hatfield2015full}.
As $\widetilde{\market}$ has two agents, it admits a unique best response sequence starting with agent $1$. We show that because all agents $J^k$ in the original market are satisfied after sequence $S^1 \ldots, S^{k}$, the offers $\offers^{k}$ for trades $\Omega^{1,2}$ after this sequence are the same as the offers after the $k$th best response in the merged market. This implies that the length of the best response sequence for $\widetilde{\market}$ is the number of subsequences of $S$. As two-agent markets converge by assumption, $S$ must thus converge after finitely many best responses,  concluding the proof sketch.

Our market dynamic selects unsatisfied agents uniformly at random, so it will eventually select a terminating best response sequence if one exists. So \cref{proposition:convergence-guarantees,proposition:two-agents-convergence} imply that $1$-sparse markets (even with non-FS agents) and $2$-sparse FSMs converge almost surely. Similarly, under the assumption that \cref{conjecture:two-agent-convergence} holds, every FSM converges almost surely.

\begin{theorem}
\label{theorem:market-convergence}
\ 
\begin{enumerate}[(i)]
\item Every $1$-sparse market converges almost surely to equilibrium, even when agents are not FS.
\item Every $2$-sparse FSM converges almost surely to equilibrium. 
\item Every FSM converges almost surely to equilibrium if \cref{conjecture:two-agent-convergence} holds.
\end{enumerate}
\end{theorem}

\section{Experiments}
\label{section:experiments}

In the previous section, we proved theoretical convergence guarantees for $1$-sparse and $2$-sparse FSMs, and conjectured that all FSMs converge. We now explore the market dynamic computationally, analyzing paths and time to convergence under various market settings.  Our primary objective for running experiments is to support the theoretical claims of the previous section and study the rate of convergence to equilibrium. Additionally, we analyze the resulting welfare based on agent composition, and the impact of exogenous shocks on the resulting equilibria and convergence.

\textbf{Experiment Configuration.} We assume one kind of good, and that sellers and buyers have unit supply/demand for this good, i.e., they want to buy/sell at most one item of this good\footnote{Note that neither of these are requirements of the market setting, but allow for a clear evaluation.}. Each buyer and seller's value $c_i \in C$ for an item of this good is drawn uniformly at random from a predetermined set of integers, $C = \{1, \ldots, 100\}$. For the seller, the value can be understood as the production cost, and for a buyer, their perceived fundamental value of the good. The valuation function of a buyer $i$ is given by $v^i(\emptyset) = 0$, $v^i(\Psi) = c_i$ if $|\Psi \cap \Omega_i| = 1$ and $v^i(\Psi) = -\infty$ otherwise. For sellers, we have $v^i(\emptyset) = 0$, $v^i(\Psi) = -c_i$ if $|\Psi \cap \Omega_i| = 1$ and $v^i(\Psi) = -\infty$ otherwise. Intermediaries cannot sell more than they buy, and have no value for retaining items, so their valuation is $v^i(\Psi) = 0$ if $\Psi_i$ contains the same number of buying and selling trades, and $v^i(\Psi) = -\infty$ otherwise.  Environments are configured in Phantom \citep{ardon2023phantom}.

\paragraph{Market Topologies.}
We analyze several FSM topologies, ranging from bipartite networks of buyers and sellers to general networks with buyers, sellers, and intermediaries.

\begin{description}
    \item[BS:] Bipartite networks consisting of buyers and sellers. Buyers only trade with sellers and vice versa. 
    \item[BIS:] Networks with buyers, sellers, and intermediaries. Buyers and sellers only trade through intermediaries, and intermediaries do not trade with other intermediaries.
    \item[General:] Networks with buyers, sellers, and intermediaries. Buyers and sellers transact through intermediaries or directly. Additionally, intermediaries may trade with other intermediaries, and cycles are permitted.
\end{description}

Examples of these topologies, and their construction processes, are described and visualized in the appendix. The BS topology models direct trade between buyers and sellers, e.g., direct-to-consumer or business-to-consumer markets. The other topologies, which include intermediaries, can model a number of more complex markets, such as used car markets with buyers and sellers interacting through dealers \cite{hatfield2013stability}, energy markets with consumers and generators mediated through suppliers \cite{morstyn2018bilateral}, or even general supply chains \cite{ostrovsky2008stability}.

\subsection{Experimental Results}
The convergence rates across the network structures with different market \textit{sizes} (number of market participants $N$), and \textit{compositions} (the proportion of buyers, sellers, and intermediaries) are analyzed. Each configuration is run 100 times on a standard personal computer\footnote{Run with Python 3 on a 2022 Macbook Air M2 with 8GB RAM.}, with the mean $\pm$ standard deviation presented.

\paragraph{BS Networks.} For each buyer/seller pair $i,j$ in a BS network, a potential trade between $i$ and $j$ is added with some probability $r$ (here, $r=0.1$).

Convergence paths are visualized in \cref{figConvergence}. Increasing the number of agents increases the time to convergence, as there is more room for bargaining/trade among the agents (e.g., a buyer has access to more sellers, so has more negotiation opportunity). Importantly, in each case, we see convergence (the satisfied rate converges to $1$). Convergence rates across $N$ are visualized in \cref{figConvergenceRatesBS}. 

Varying the proportion of buyers in the market (\cref{figBSWelfare}), we recover the expected economic welfare behavior that as the proportion of buyers increases (decreases), their welfare, measured by the average utility, decreases (increases), due to market competition, resulting in reduced consumer surplus (increased producer surplus).

\begin{figure*}[!tb]
\centering
\begin{subfigure}[b]{0.33\textwidth}
\centering
\includegraphics[width=\textwidth]{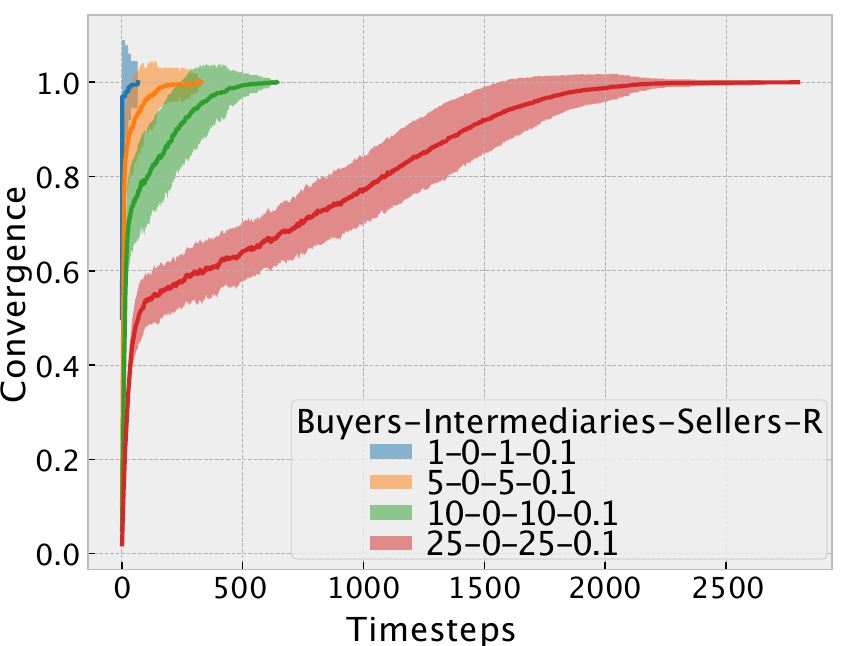}
\caption{BS}
\label{figConvergence}
\end{subfigure}
\begin{subfigure}[b]{0.33\textwidth}
\centering
\includegraphics[width=\textwidth]{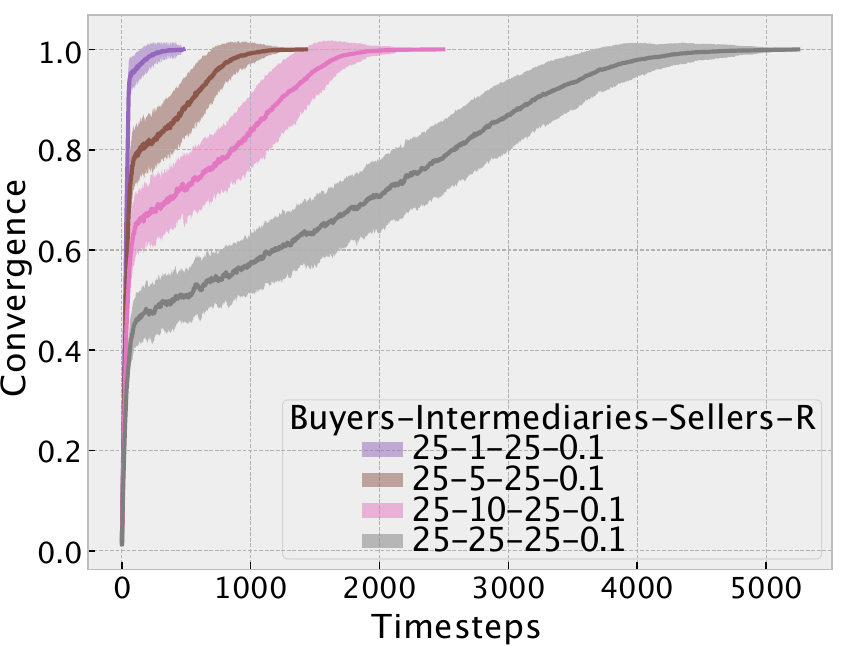}
\caption{BIS}
\label{figintermediaryConvergence}
\end{subfigure}
\begin{subfigure}[b]{0.33\textwidth}
\centering
\includegraphics[width=\textwidth]{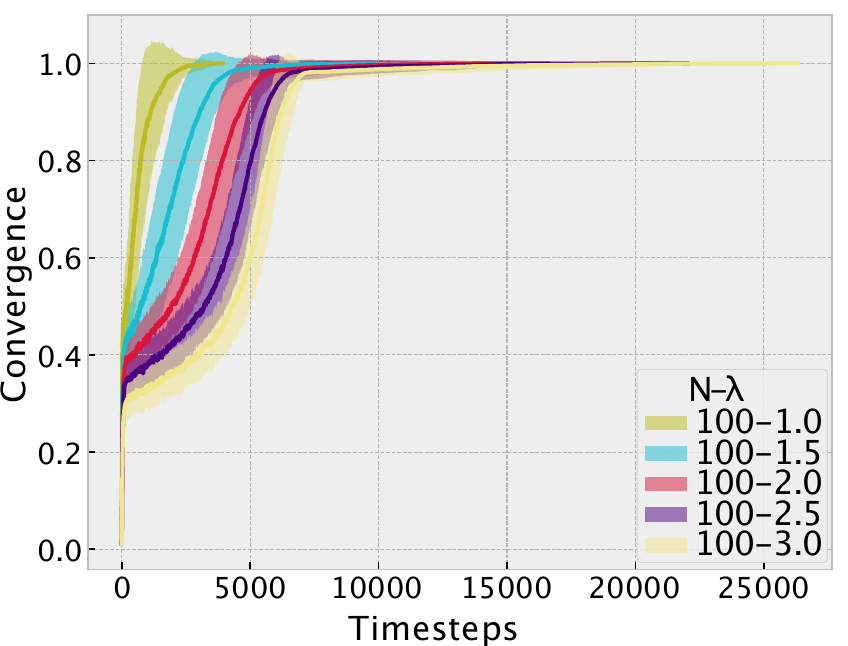}
\caption{General}
\label{figGeneralConvergence}
\end{subfigure}
\caption{Convergence paths by market topology. The $x$-axis shows the number of iterations (e.g., number of best response updates), and the $y$-axis shows the proportion of \textit{satisfied} agents. 
}
\end{figure*}

\begin{figure*}[!tb]
    \centering
    \includegraphics[width=.72\textwidth]{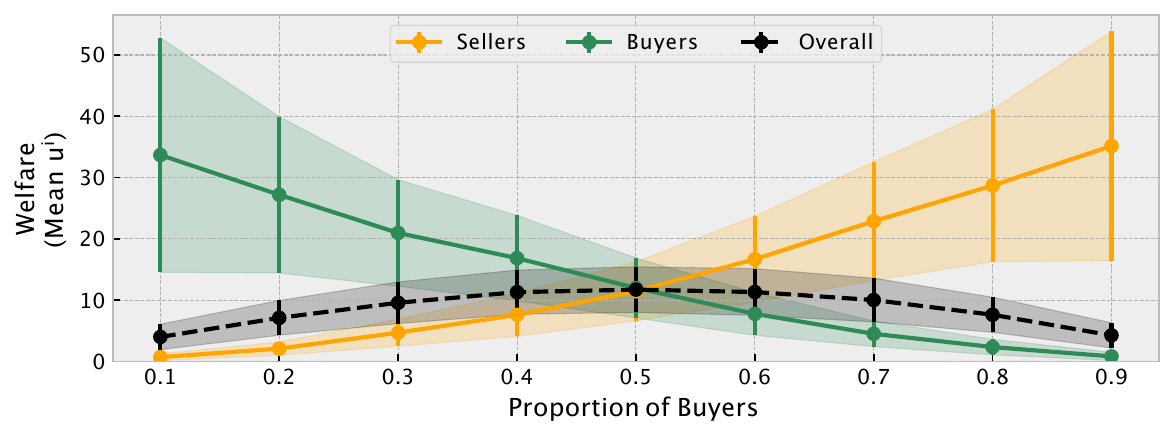}
    \caption{Social welfare (mean $u_i$) by agent composition in the BS network. The $x$-axis varies the proportion of buyers, and the $y$-axis displays the resulting welfare.}
    \label{figBSWelfare}
\end{figure*}

\paragraph{BIS Networks.} Introducing intermediaries into the network alters the convergence rates as transactions must happen through an intermediary. We explore different numbers of intermediaries, again, with a fixed connectivity rate of $r=0.1$, where for the possible pairs $i, j$ at least one of $i$ or $j$ must be an intermediary (e.g., there is no direct connection between buyers and sellers).

Convergence rates across an increasing number of intermediaries are visualized in  \cref{figintermediaryConvergence}. In each case, the market converges to an equilibrium, as predicted by the theory. Increasing the number of intermediaries increases the time to convergence, as there are more opportunities for trade. When the number of intermediaries is low, buyers/sellers have little room for negotiation, so convergence is more rapid. The more intermediaries in the market, the lower their utilities are due to the competition, which results in higher utilities for the buyers and sellers.

\paragraph{General Networks.}
We now consider general trading networks. For these experiments, we generate Erdős-Rényi networks
with connectivity rates $r=\frac{\lambda}{N}$, where $N = 100$ and $\lambda = \{1, 1.5, 2, 2.5, 3\}$.

Convergence rates across $\lambda$ are visualized in \cref{figGeneralConvergence}. Despite featuring cycles, these trading networks still converge, at a rate approximately linear with $\lambda$ (\cref{figConvergenceRatesGeneral}). While we do not yet have theoretical guarantees of convergence under this setting, these results help to show that experimentally these general topologies also converge, opening a promising line of future research building on \cref{theorem:market-convergence}. The welfare of the agents does not change drastically across $\lambda$, but generally, as $\lambda$ increases the proportion of intermediaries in the network increases, so the welfare for buyers/sellers increases, and the welfare for intermediaries decreases (see appendix).

\subsection{Shocks}
\begin{figure*}[!htb]
\centering
\begin{subfigure}[b]{0.41\textwidth}
\centering
\includegraphics[width=\textwidth]{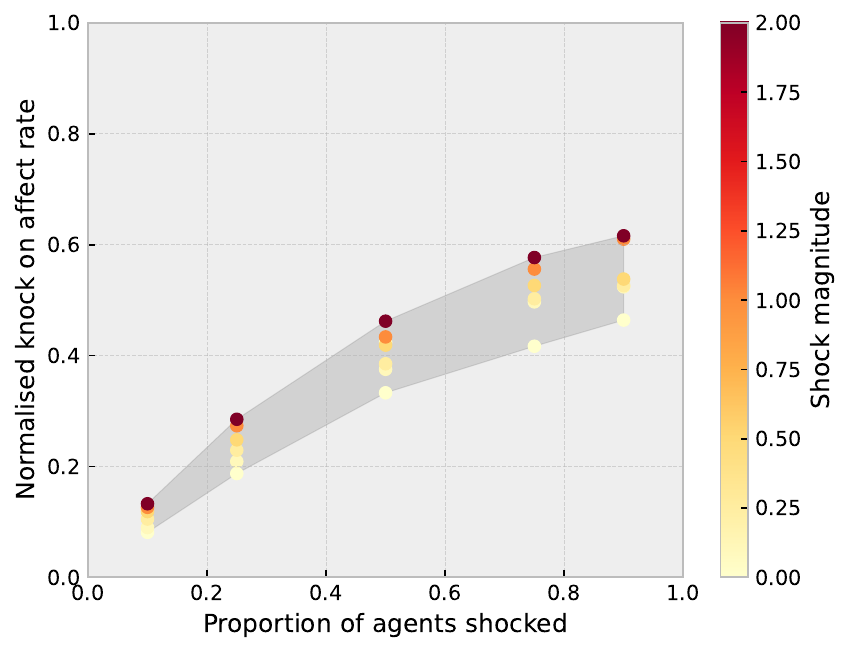}
\caption{Propagation}\label{figAdditionalImpact}
\end{subfigure}
\hspace{2em}
\begin{subfigure}[b]{0.41\textwidth}
\centering
\includegraphics[width=\textwidth]{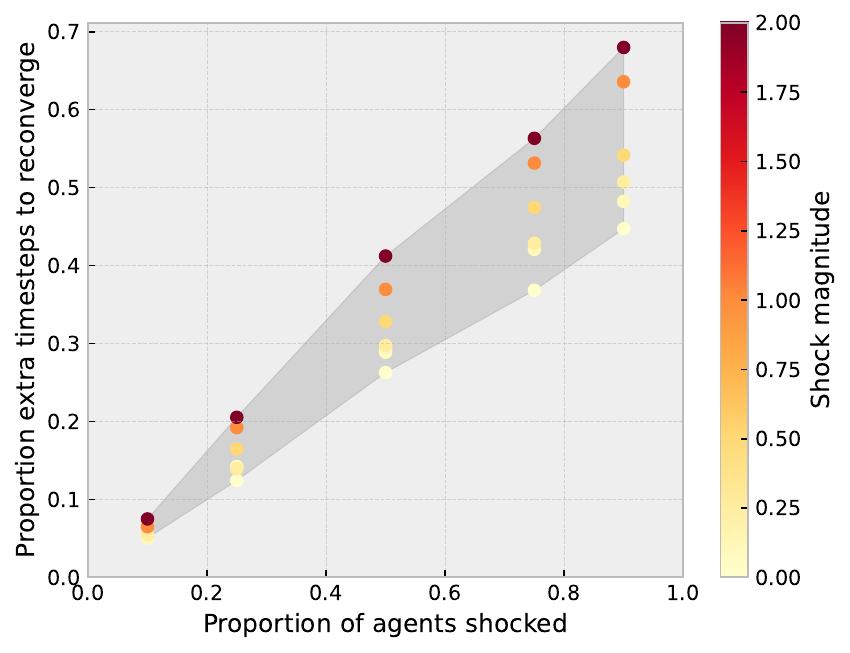}
\caption{Reconvergence speed}\label{figReconvergence}
\end{subfigure}
\caption{The impact of exogenous shocks in a BS network. \cref{figAdditionalImpact} shows the propagation of shocks measured as the additional proportion of unsatisfied agents. \cref{figReconvergence} shows the normalized reconvergence speed (as a proportion of the original convergence speed). Each $\circ$ represents the mean impact for a shock size (darker color = larger magnitude). The filled region shows the effect range across these shock sizes.}
\label{figShocks}
\end{figure*}

To understand how exogenous shocks impact the market, shocks are applied to an agent's value $c_i$, resampling to a new value in $C$. Shocks have a specific size $s$, which controls the variation from the original valuation, i.e., a shock size of $0.1$ means the updated $c_{i,t}$ will be in the range $[c_{i,t-1} * 0.9, c_{i,t-1} * 1.1]$ (clipped to the values in $C$). Shocks are applied to different proportions of agents in the network to analyze their propagation. Shocks only occur after the dynamics have converged to understand how these perturbations affect an otherwise converged system.

The results for the BS network are visualized in \cref{figShocks}. We look at two key quantities: the proportion of additional (non-shocked) agents affected (\cref{figAdditionalImpact}), e.g., due to knock-on effects of one borrower updating their valuation, and the normalized time taken to reconverge (\cref{figReconvergence}).

\textbf{Propagation.} When considering the propagation of shocks (\cref{figAdditionalImpact}), a knock-on effect is observed in all cases, with additional non-shocked agents becoming affected (i.e.,becoming unsatisfied and needing to update their prices). For example, when applying a shock to $25\%$ of the agents, approximately $50\%+$ of the agents become impacted and must update their responses. Generally speaking, the larger the shock magnitude, the more propagation throughout the network is observed, agreeing with real world insights on shock propagation's, for example in production networks \cite{huneeus2018production} and global trade-investments \cite{Starnini2019}, where the magnitude of the shock determines its impact on the overall system.

\textbf{Reconvergence.}
The additional time to converge (\cref{figReconvergence}) increases approximately linearly with the proportion of shocked agents and the shock magnitude, showing the disruption effects. As expected by the theory, irrespective of the shock size, the dynamics eventually settle down and reconverge to a new equilibrium. This reconvergence is significantly faster than the original convergence speed (\cref{figReconvergence}), essentially serving as a warm start to the convergence process, which is particularly true for small shock magnitudes. 

\subsection{Results Summary}
These experimental results support and extend the theoretical guarantees in \cref{section:theory}, by simulating the market dynamic for various trading networks and establishing consistent convergence across market topologies. The best response sequences constructed in the proof of \cref{proposition:convergence-guarantees} can be exponentially long. By contrast, we show experimentally that these markets converge rapidly. Additionally, we demonstrate that more complex network structures also still converge under the proposed market dynamic, opening a promising line of future analysis extending \cref{proposition:convergence-guarantees}. The experimental analysis provided insights into the resulting welfare and shock impacts across various market topologies, expanding the understanding of such market processes.

\section{Discussion and Conclusion}
\label{section:conclusion}
We formulate a decentralized market dynamic in a heterogeneous network of trading agents. Our dynamic captures iterative negotiation among the agents, by allowing buyers (sellers) to refuse trades via making counteroffers that are slightly lower (higher) than the price being offered. We demonstrate theoretically and experimentally how prices converge to an equilibrium via an uncoordinated negotiation (from arbitrary initial prices), providing the first such analysis in general market settings.

Theoretically, we develop a reduction that ensures convergence of many-agent markets iff two-agent markets converge. We apply this reduction to prove that our market dynamic reaches equilibrium for markets with tree topologies, as well as 2-sparse FSMs.
Experimentally, we provide empirical evidence and extensions to the convergence guarantees and convergence process, demonstrating that fully-substitutable markets converge to an equilibrium and do so in a manner significantly faster than suggested by the proofs in our theoretical work. Additionally, we highlight several natural and desirable features of the market, e.g., market reactions to exogenous shocks and the welfare impact on agents based on market compositions.

In addition to proving our \cref{conjecture:two-agent-convergence} about the convergence of two-agent markets with many trades, future work could address more sophisticated choices of counteroffers that may lead to faster convergence to equilibrium, e.g., modifying the step size $\varepsilon$ in the spirit of adaptive stochastic gradient based optimization algorithms \citep{DBLP:journals/corr/KingmaB14, duchi2011adaptive}. Additionally, the agents in our market dynamic optimize their immediate utility, but future iterations could also model strategic traders \citep{vadori2024towards} who consider the long-term impact of their offers.

\section*{Acknowledgments}
Goldberg and Lock were supported by a JP Morgan faculty fellowship and EPSRC grant EP/X040461/1.

\section*{Disclaimer}
This paper was prepared for informational purposes in part  by the Artificial Intelligence Research group of JPMorgan Chase \& Co. and its affiliates ("JP Morgan'') and is not a product of the Research Department of JP Morgan. JP Morgan makes no representation and warranty whatsoever and disclaims all liability, for the completeness, accuracy or reliability of the information contained herein. This document is not intended as investment research or investment advice, or a recommendation, offer or solicitation for the purchase or sale of any security, financial instrument, financial product or service, or to be used in any way for evaluating the merits of participating in any transaction, and shall not constitute a solicitation under any jurisdiction or to any person, if such solicitation under such jurisdiction or to such person would be unlawful.

\bibliography{bibliography}

\appendix

\clearpage

\section{Convergence to Equilibrium}
Here we develop the formal statements and proofs for our theoretical results on market convergence.

\subsection{Two-agent markets}
Consider markets with two agents, and at most two trades between them. In a two-agent market, the market sequence must alternate between the two agents, and so is unique.
We prove that the number of best responses until a market with one trade converges to equilibrium is $O(V)$, and two-agent FSMs with $2$ trades converge after $O(V^2)$ best responses. Here we recall that $V$ represents an upper bound on the absolute values of both agents’ valuation functions and initial offers. We begin by showing that the two-agent market dynamic either terminates or cycles, and any such cycle must consist of an even number of best responses.
Without loss of generality, we set the best response dynamic's discrete step size $\varepsilon$ to $\varepsilon = 1$ in the following.

\begin{lemma}
\label{lemma:offers-are-bounded}
\label{lemma:boxing-lemma}
At any point during the market dynamic, we have $-2V+1 \leq \offer^{b(\omega)}_{\omega}, \offer^{s(\omega)}_{\omega} \leq 2V+1$ for any $\omega \in \Omega$.
So the best response dynamic either terminates, or the sequence of best responses cycles.
Moreover, cycles in two-agent markets have even length.
\end{lemma}
\begin{proof}
The condition $-2V-1 \leq \offer^{b(\omega)}_{\omega}, \offer^{s(\omega)}_{\omega} \leq 2V + 1$ holds by definition of $V$ for the initial offers of the dynamic.

We now show that no offer can increase to a value greater than $2V+1$ at any point in the best response dynamic. (The argument that no offer can decrease to a value less than $-2V-1$ is analogous.) Fix some trade $\omega \in \Omega$. For $\offer^{b(\omega)}_{\omega}$ to increase to a value greater than $2V+1$, the seller must at some point set $\offer^{s(\omega)}_{\omega}$ to $2V+1$, the buyer must then also set their corresponding offer $\offer^{b(\omega)}_{\omega}$ to $2V+1$, and then the seller eventually sets their offer to $2V+2$. We now show that the buyer will not do this.

Suppose $\offers$ are the offers immediately after the seller of trade $\omega$ sets her offer $\offer^{s(\omega)}_{\omega}$ to $2V+1$. Assume for a contradiction that the buyer sets their offer to $\offer^{b(\omega)}_\omega = 2V+1$, so he demands a bundle $\Phi$ with $\omega \in \Phi$ at $\offers$. The difference of his utility for bundles $\Phi$ and $\Phi \setminus \{ \omega \}$ at $\offers$ is $v^{b(\omega)}(\Phi) - v^{b(\omega)}(\Phi \setminus \{\omega\}) - \offer^{s(\omega)}_{\omega} \leq 2V - 2V - 1 < 0$, contradicting the fact that he demands bundle~$\Phi$.

By the above, the number of different combinations of offers between $-2V-1$ and $2V+1$ that an agent can make is finite. So in any best response sequence that does not terminate, an agent eventually repeats herself by making the same offers as in a previous best response. As an agent's best response depends only on her counterparts' current offers (so she is `memory-less'), the sequence cycles after the first time an agent repeats herself.

Finally, if the market has two agents, the fact that the sequence alternates between the two agents implies that the cycle is even.
\end{proof}

Our proof makes use of the following intuitive notation to represent best response sequences for two-agent markets. (Recall that such sequences alternate between the two agents.) Each such sequence is associated with two `parallel' \textit{rows}, one for each trade, that consist of (finite or infinite) strings of characters $A$ and $R$. (See the example below). Each of the \textit{columns} in this sequence corresponds to a best response by an agent: an $A$ means that the agent accepts her counterpart's offer (and thus sets her offer to match her counterpart's), and an $R$ means that the agent rejects her counterpart's offer (and thus sets her offer to be off by $\varepsilon$, as defined in \cref{definition:best-response}). The odd columns correspond to the agent who initiates the sequence, whilst even columns correspond to the other agent. As an example, we give the cyclic best response sequence of \cref{example:complements-cycle} between a buyer (b) and a seller (s):

\begin{center}
\begin{tabular}{llllllll}
\small
& s & b & s & b \\
\midrule
$\omega$ & $R$ & $A$ & $A$ & $R$ \\
$\varphi$ & $R$ & $R$ & $A$ & $A$
\end{tabular}
\end{center}
In her first best response, the agent rejects both the buyer's offers, so she sets her offer to one higher than the buyer's offers. In the buyer's subsequent best response, he matches the seller's offer for trade $\omega$ and sets an offer one less than the seller's offer for trade $\varphi$.
From this notation, we can also see straightforwardly that this sequence cycles: in both rows, the buyer and seller each change their offers once, so the offers before and after the four best responses is identical.

We will omit the agent and trade designations below, due to symmetry arguments.



Suppose the market cycles. Then we can ignore the non-cyclic beginning of the string representation, and focus entirely on the cycle at the end. A substring of a row is an \textit{$A$-string} if it consists only of $A$s, and an $A$-string is \textit{maximal} if it is not preceded or followed by any additional $A$s beyond the string. (Maximal) $R$-strings are defined likewise.

We observe structural properties of the rows of a cycle.

\begin{lemma}
\label{lemma:consecutive-BRs}
The market dynamic terminates iff its BR sequence contains two identical consecutive columns.
\end{lemma}
\begin{proof}
Suppose two consecutive columns are identical, i.e., consecutive best responses of the two agents are the same for both trades. After these best responses, no agent will change their offers any more. So both agents are satisfied and the market dynamic terminates. Conversely, suppose the market dynamic terminates. Then the last best response must leave all offers unchanged, and match the demand of the penultimate best response. So the result follows.
\end{proof}

\begin{lemma}
\label{lemma:even-strings}
Every row of the cycle contains a maximal $A$-string or maximal $R$-string of even length.
\end{lemma}
\begin{proof}
Fix a row associated with trade $\omega$. We refer to the buyer of $\omega$ as the \textit{buyer} and the seller of $\omega$ as the \textit{seller}. 

Suppose for a contradiction that the row consists of alternating $A$-strings and $R$-strings of odd length only. As the cycle has even length, we have at least one of each such string. Moreover, each $A$-string is initiated by the same agent, and likewise each $R$-string is initiated by the other agent. Without loss of generality, suppose the $A$-strings are initiated by a buyer. (The case with a seller is analogous.)

Consider any such $A$-string in the row. It is preceded and followed by at least one $R$. As the buyer initiates this string, her offer increases after her first best response of the string and remains the same after all her subsequent best responses of the string. The seller's offers remain unchanged after all his best responses of the string. Secondly, consider any maximal $R$-string in the row, which has odd length and is initiated by the seller. The seller's offer increases after his first best response, and remains the same for all subsequent best responses, while the buyer's offers are unchanged after all her best responses.

As the row contains at least one $A$-string and one $R$-string, the buyer and seller's offers for trade $\omega$ either both strictly increase or both strictly decrease across iterations of the cycle. But this contradicts \cref{lemma:boxing-lemma}, so the row must contain an even-length $A$-string or $R$-string.
\end{proof}

We now show the first statement of \cref{proposition:two-agents-convergence}: the best response dynamic for a market with single trade $\omega$ terminates after $O(V)$ best responses. 

\begin{proof}[Proof of the first part of \cref{proposition:two-agents-convergence}]
Suppose the market cycles. This cycle consists of a single row. \Cref{lemma:consecutive-BRs} implies that the row must consist of alternating $A$s and $R$s, while \cref{lemma:even-strings} tells us that the row must have at least two consecutive $A$s or $B$s, a contradiction.
\end{proof}

Now consider two-agent markets with two trades. We develop the proof that these markets terminate for the case that both trades share the same buyer and seller (so one agent is the buyer of both trades, and the other agent is the seller). The case with two `crossing' trades is proved analogously.

\begin{observation}
\label{observation:mixed-row}
Every row of a cycle for a market with two agents and two trades contains both $A$s and~$R$s.
\end{observation}
\begin{proof}
If one row consists only of $A$s, then the other row must contain alternating $A$s and $R$s by \cref{lemma:consecutive-BRs}, contradicting \cref{lemma:even-strings}. 
\end{proof}

As agents have fully-substitutable quasi-linear preferences, we can exclude the possibility of certain three-column fragments showing up in any cycle.
\begin{lemma}
\label{lemma:two-trade-fragments}
The cycle of a market with a buyer and seller, and two trades does not contain three-column fragments
\begin{center}
\small
\begin{tabular}{lll}
$A$ & $A$ & $R$ \\
$A$ & $R$ & $*$
\end{tabular}
\hspace{1.5em}
\begin{tabular}{lll}
$R$ & $A$ & $A$ \\
$R$ & $R$ & $*$
\end{tabular}
\hspace{1.5em}
\begin{tabular}{lll}
$A$ & $A$ & $R$ \\
$*$ & $A$ & $R$
\end{tabular}
\hspace{1.5em}
\begin{tabular}{lll}
$R$ & $R$ & $A$ \\
$A$ & $R$ & $A$
\end{tabular}
\hspace{1.5em}
\hspace{1.5em}
\begin{tabular}{lll}
$R$ & $A$ & $R$ \\
$R$ & $R$ & $A$
\end{tabular}
\hspace{1.5em}
\end{center}
or the fragments obtained by swapping the two rows. The fragments can be initiated by either agent and $*$ denotes any of $A$ or~$R$.
\end{lemma}
\begin{proof}
Let $\omega$ and $\varphi$ be the trades associated with the first and second row, respectively. We assume that the first and last column correspond to the seller, and the middle column corresponds to the buyer. (The converse case is symmetric.)

Consider the first fragment. The buyer's offer for $\omega$ remains unchanged and her offer for $\varphi$ reduces by $1$ so, by substitutability, the seller's demand of trade $\omega$ cannot reduce. But the seller initially demands $\omega$ and subsequently does not demand $\omega$ any more, a contradiction. The impossibility of the second fragment follows analogously by substitutability.

Now consider the third fragment. The buyer's offer $\offer_{\omega}$ for trade $\omega$ does not change after her best response, so by definition of demand we see $v^s(\{ \omega \}) + \offer_{\omega} > v^s(\emptyset)$ from the seller's first best response and $v^s(\{ \omega \}) + \offer_{\omega} < v^s(\emptyset)$ from the seller's second best response, a contradiction. The impossibility of the remaining fragments follows analogously by quasilinearity of demand.
\end{proof}

We are now ready to prove the second statement in \cref{proposition:two-agents-convergence}: two-agent FSMs with two trades converge.

\begin{proof}[Proof of the second part of \cref{proposition:two-agents-convergence}]
Fix some market with two agents and two trades, $\omega$ and $\varphi$. Suppose that both trades have the same buyer, and thus also the same seller. (The other case is proved analogously.)
Suppose for a contradiction that the market cycles. Associate $\omega$ with the top row in our cycle representation. We already know that the cycle has even length. If it had length 2, \cref{lemma:even-strings,observation:mixed-row} would contradict each other, so the cycle has even length of at least $4$. Moreover, the top row has a maximal $A$- or $R$-string of even length. Suppose first that it has a maximal $A$-string of even length $k$. So the cycle contains the fragment
\begin{center}
\begin{tabular}{llllllll}
\small
$R$ & $A$ & $A$ & $\cdots$ & $A$ & $A$ & $R$ & $X_{k+2}$ \\
$Y_0$ & $Y_1$ & $Y_2$ & $\cdots$ & $Y_{k-1}$ & $Y_{k}$ & $Y_{k+1}$ & $Y_{k+2}$ 
\end{tabular}
\end{center}

\Cref{lemma:two-trade-fragments} rules out all possibilities for $(Y_{k-1}, Y_{k},Y_{k+1})$ apart from $(R,A,A)$. \cref{lemma:consecutive-BRs,lemma:two-trade-fragments} then tell us that $Y_l = R$ for odd $l \leq k-1$, and $Y_l = A$ for even $l \leq k$. By \cref{lemma:two-trade-fragments} (first fragment with swapped rows), $Y_{k+2}=A$ and, by \cref{lemma:consecutive-BRs}, $X_{k+2}=A$. So we get
\begin{center}
\begin{tabular}{llllllll}
\small
$R$ & $A$ & $A$ & $\cdots$ & $A$ & $A$ & $R$ & $A$ \\
$A$ & $R$ & $A$ & $\cdots$ & $R$ & $A$ & $A$ & $A$ 
\end{tabular}
\end{center}
Suppose without loss of generality that the fragment is initiated by the seller (the case with the buyer is symmetric). Let $\offers$ and $\offers'$ be the sellers' offers before the buyer's first and last best response. By construction, $\offer'_\omega = \offer_\omega$ and $\offer'_\varphi > \offer_\varphi$. The buyer demands $\{\varphi\}$ at $\offers$ and $\{\omega, \varphi \}$ at $\offers'$. But this contradicts quasilinearity of demand, so we are done.

Now suppose that the top row only has a maximal $R$-string of even length $k$, which is preceded and followed by an $A$. Without loss of generality, assume the $R$-string is initiated by a buyer. By \cref{lemma:consecutive-BRs}, the other row must alternate between $R$ and $A$ along the $R$-string. If this alternation starts with an $R$, we get fragment
\begin{center}
\begin{tabular}{lllllll}
\small
$A$ & $R$ & $R$ & $\cdots$ & $R$ & $R$ & $A$\\
$*$ & $R$ & $A$ & $\cdots$ & $R$ & $A$ & $R$ 
\end{tabular}
\end{center}
Here the last $R$ in the bottom row follows by exclusion from \cref{lemma:two-trade-fragments}. As the $R$-string is even, we see that the seller makes the same offer for trade $\omega$ before the buyer's first and last best response. So the fact that the buyer initially demands $\emptyset$ and then demands $\{\omega\}$ contradicts quasilinearity of demand. Secondly, suppose the alternation of $A$ and $R$ below the $R$-string starts with an $A$. By repeated application of \cref{lemma:two-trade-fragments} (recall that we can swap rows in \cref{lemma:two-trade-fragments}), we then get fragment
\begin{center}
\begin{tabular}{llllllll}
\small
$A$ & $R$ & $R$ & $\cdots$ & $R$ & $R$ & $A$ & $R$\\
$R$ & $A$ & $R$ & $\cdots$ & $A$ & $R$ & $R$ & $R$ 
\end{tabular}
\end{center}
Similarly to above, the buyer makes the same offer for trade $\omega$ before the seller's first and last best response. As the seller initially demands $\{\omega\}$ and then demands $\emptyset$, this contradicts quasilinearity of demand.
\end{proof}

As stated in \cref{conjecture:two-agent-convergence}, we expect that two-agent markets with any number $m$ of trades converge in $O(V^m)$ steps.

\subsection{Markets with more than two agents}

Recall that $V$ is an upper bound on the value $|v^i(\Psi)|$ for any bundle $\Psi \subseteq \Omega_i$ for any agent $i$, and $f(m,V) \leq \infty$ denotes the worst-case number of best responses required for any two-agent FSM with $m$ trades to converge.

We now turn to the proof of \cref{proposition:convergence-guarantees}. In this proof, we reduce markets in two ways, which make use of valuation function transformations introduced in \citet[Section 5]{hatfield2015full}. Importantly, these transformations preserve full substitutability. The first reduction, which restricts a given market $\market$ to a subset of agents, allows us to construct a specific best response sequence for $\market$. The second reduction, which reduces $\market$ to a two-agent market, then allows us to prove that this best response sequence converges.

\paragraph{Restricted markets.}
Fix market $\market = (I, \Omega, v)$ in state $(U, \offers)$, and a subset $J \subseteq I$ of the market's agents. A trade in $\Omega$ is \textit{internal} if its buyer and seller are both in $J$, and \textit{external} if exactly one is in $J$. For any $\Psi \subseteq \Omega$, we let $\Psi^{int}$ and $\Psi^{ext}$ denote its interior and exterior trades.

We define the \textit{restricted market} $\widehat{\market}(J, \offers) = (J, \Omega^{int}, \widehat{v})$ obtained from $\market$ and offers $\offers$ by restricting the market to agents~$J$ and endogenizing the external trades $\omega \in \Omega^{ext}_i$ of each external agent $i$ at fixed prices $\offer^{-i}_{\omega}$. To achieve this, we define the valuation $\widehat{v}^i$ of each agent $i \in J$ for each $\Phi \subseteq \Omega^{int}_i$ as
\begin{equation}
\label{eq:restricted-valuation}
\widehat{v}^i(\Phi) \coloneqq 
\max_{\Psi \subseteq \Omega^{ext}_i} \left [ v^i(\Phi \cup \Psi) - \sum_{\omega \in \Psi} \chi^i_{\omega}\offer^{-i}_{\omega} \right ].
\end{equation}
These valuations implicitly endow each agent $i \in J$ with the opportunity to purchase (any subset of) the trades in $\Omega^{ext}_i$ at `frozen prices' $\offer^{-i}_{\omega}$, their counterparts' offers, if this maximizes their utility; even though these trades do not appear in the market itself. For internal agents, $\widehat{v}^i = v^i$ as $\Omega^{ext}_i = \emptyset$. We note that the bundle $\Psi$ maximizing \eqref{eq:restricted-valuation} is unique because of the way we perturb valuations to break ties. Moreover, \citet[Theorem 6]{hatfield2015full} show that $\widehat{v}^i$ is fully-substitutable as long as $v^i$ is. We also map the state $(U, \offers)$ of the original market to state $(\widehat{U}, \widehat{\offers})$ of the restricted market by letting $\widetilde{U} = U \cap \Omega^{int}$ and $\widehat{\offers}$ be the offers $\offers$ restricted to the trades $\Omega^{int}$.

We now show that if we apply the same best response sequence with agents from $J$ to $\market$ in state $(U, \offers)$ and $\widehat{\market}(J,\offers)$ in state $(\widehat{U}, \widehat{\offers})$, then the same subset of agents $J$ is satisfied at the end. In particular, this means that if there exists a terminating best response sequence for the restricted market, then there exists a best response sequence for the original market after which all agents in $J$ are satisfied.

\begin{lemma}
\label{lemma:restricted-market}
Fix market $\market$ in state $(U, \offers)$. If $\widehat{\market}(J,\offers)$ in state $(\widehat{U}, \widehat{\offers})$ terminates after some best response sequence, then applying the same best response sequence in market $\market$ in state $(U, \offers)$ results in all agents $J$ being satisfied.
\end{lemma}
\begin{proof}
Fix a best response sequence for $\widehat{\market}(J,\offers)$ in state $(\widehat{U}, \widehat{\offers})$. So this sequence contains only agents from set $J$. We will now argue that, after the best response of the $k$-th agent, i) the current offers for every trade $\omega \in \Omega^{int}$ are identical in both markets; and ii) the same subset of $J$ is satisfied in both markets $\market$ and $\widehat{\market}(J,\offers)$.

We argue by induction on $k$ for the first statement. Let $i \in J$ be the agent performing the $k$th best response, and assume that the statement holds for the offers $\offers$ and $\widehat{\offers}$ in the respective markets immediately prior to the $k$th best response. (It is true by definition for $k=1$.) The agent's utility for any bundle $\Theta \subseteq \Omega_i^{int}$ at $\widehat{\offers}$ in market $\widehat{\market}$ (with valuation $\widehat{v}^i$) is
\[
\widehat{v}^i(\Theta) - \sum_{\omega \in \Theta} \chi^i_\omega \offer^{-i}_{\omega} = v^i(\Theta \cup \Psi) - \sum_{\omega \in \Theta \cup \Psi} \chi^i_{\omega}\offer^{-i}_{\omega},
\]
for appropriate $\Psi \subseteq \Omega_i^{ext}$, by~\eqref{eq:restricted-valuation}. This is also her utility for bundle $\Theta \cup \Psi$ at $\offers$ in $\market$ with valuation $v^i$. So this implies that she demands bundle $\Phi$ at $\offers$ in market $\market$ if and only if she demands $\Phi^{int}$ at $\widehat{\offers}$ in the market $\widehat{\market}$. Thus her best response in both markets will make the same offers for all trades $\Phi^{int}$.

For the second statement, recall that an agent becomes satisfied when she best responds, so agents in $J$ become satisfied at the same time in both markets. Secondly, an agent in $J$ changes from satisfied to unsatisfied only if one of her trading partners in $J$ modifies his offers. But the first statement tells us that the offers for trades $\Omega^{int}$ change identically in both markets, so agents also become unsatisfied at the same time in both markets.
\end{proof}

\paragraph{Merging agent sets.}
Next, fix market $\market = (I, \Omega, v)$ with at least two agents, and partition $I$ into non-empty sets $J^1$ and $J^2$.
Moreover, for both $k \in \{1,2\}$, define $\Omega^k = \{ \omega \in \Omega \mid b(\omega), s(\omega) \in J^k \}$ as the set of trades in $\Omega$ with both endpoints in $J^k$, and let $\Omega^{1,2}$ be all other trades, with one endpoint in $J^1$ and the other endpoint in $J^2$.
We now construct a two-agent market $\widetilde{\market}(J^1, J^2) = (\{1,2\}, \widetilde{\Omega}^{1,2}, \widetilde{v})$ from $\market$ in which $J^1$ and $J^2$ are merged into single agents, denoted $1$ and $2$.

For each trade $\omega \in \Omega^{1,2}$ of $\market$, we write $\widetilde{\omega}$ for the same trade with the endpoint in $J^k$ replaced by agent $k$. For any set $\Psi \subseteq \Omega^{1,2}$, we define $\widetilde{\Psi} \coloneqq \{ \widetilde{\omega} \mid \omega \in \Psi \}$, so in particular there is a natural one-to-one correspondence between the trades $\Omega^{1,2}$ of market $\market$ and the trades $\widetilde{\Omega}^{1,2}$ of the two-agent market. Finally, the valuation $\widetilde{v}^k$ of each agent $k \in \{1,2\}$ is
\begin{equation}
\label{eq:merged-valuation}
\widehat{v}^k(\widetilde{\Phi}) \coloneqq \max_{\Psi \subseteq \Omega^k} \sum_{i \in J^k} v^i(\Phi \cup \Psi), \quad \forall \widetilde{\Phi} \subseteq \widetilde{\Omega}^{1,2}.
\end{equation}
This definition captures the notion of a merger of $J^k$, as the agents in this set can conduct internal trades for free (as the cost of buyer and the income of the seller cancel out). So the set of agents $J^k$ chooses a set $\Psi \subseteq \Omega^k$ of `interior' trades that maximizes its aggregate valuation function. See also \citet[Section 5.2]{hatfield2015full}, where the valuations of merged agents is first introduced. In particular, \citet[Theorem 7]{hatfield2015full} show that merging multiple fully-substitutable agents results in a fully-substitutable agent, so agents $1$ and $2$ in the two-agent market are both fully-substitutable.

For any set $U$ of unsatisfied agents for $\market$, let $\widetilde{U} \subseteq \{1,2\}$ with $k \in \widetilde{U}$ iff $U \cap J^k \neq \emptyset$. Moreover, for any offers $\offers$ for $\market$, we define offers $\widetilde{\offers}$ for $\widetilde{\market}(J^1, J^2)$ by $\widetilde{\offer}^{b(\widetilde{\omega})}_{\widetilde{\omega}} \coloneqq {\offer}^{b({\omega})}_{{\omega}}$ and $\widetilde{\offer}^{s(\widetilde{\omega})}_{\widetilde{\omega}} \coloneqq {\offer}^{s({\omega})}_{{\omega}}$ for all $\omega \in \Omega^{1,2}$.
We now show that a best response sequence involving agents from $J^k$ starting with market $\market$ in state $(U, \offers)$ ends with the same offers for $\Omega^{1,2}$ as if we let the agent $k$ best respond in the two-agent market in state $(\widetilde{U}, \widetilde{\offers})$.

\begin{lemma}
\label{lemma:merged-market}
Suppose $\market$ in state $(U, \offers)$ admits a BR sequence of agents from $J^k \subseteq I$ after which all agents in $J^k$ are satisfied, and let $\offers'$ be the offers after this sequence.
Then $\widetilde{\offers}'$ are the offers after a single best response from agent~$k$ in the market $\widetilde{\market}(J^1, J^2)$ in state $(\widetilde{U}, \widetilde{\offers})$.
\end{lemma}
\begin{proof}
Let $\widetilde{\Phi}$ be the bundle that agent $k$ in market $\widetilde{\market}(J^1,J^2)$ demands at $\widetilde{\offers}'$, for some $\Phi \subseteq \Omega^{1,2}$. The agent's utility for $\widetilde{\Phi}$ at $\widetilde{\offers}'$ is, by construction of the two-agent market, given by
\begin{align}
\label{eq:single-agent-utility}
    &\sum_{i \in J^k} v^i(\Phi \cup \Psi) - \sum_{\widetilde{\omega} \in \widetilde{\Phi}} \chi^k_{\widetilde{\omega}} (\widetilde{\offer}')^{-k}_{\widetilde{\omega}} \nonumber \\
    = &\sum_{i \in J^k} \left [ v^i(\Phi \cup \Psi) - \sum_{\omega \in \Phi}\chi^i_{\omega} (\offer')^{-i}_{\omega} \right ],
\end{align}
for the bundle $\Psi \subseteq \Omega^k$ that uniquely maximizes \eqref{eq:merged-valuation}.

As the agents $J^k$ are satisfied after the best response sequence, the two agents associated with each trade in $\Omega^k$ either both demand, or both don't demand, the trade. So let $\Xi \subseteq \Omega$ be the bundle of trades aggregately demanded by agents $J^k$. As the payments for the trades in $\Omega^k$ between buyer and seller cancel out, the aggregate utility of agents $J^k$ in $\market$ for $\Xi$ at $\offers'$ is
\begin{equation}
\label{eq:aggregate-utility}
\sum_{i \in J^k} \left [ v^i(\Xi) - \sum_{\omega \in \Xi \cap \Omega^{1,2}}\chi^i_{\omega} (\offer')^{-i}_{\omega} \right ].
\end{equation}
As the respective bundles $\widetilde{\Phi}$ and $\Xi$ are utility-maximizing, it follows from \eqref{eq:single-agent-utility} and \eqref{eq:aggregate-utility} that ${\Phi} = \Xi \cap \Omega^{1,2}$.

We now use this fact to prove the main claim: $(\widetilde{\offer}')^{b(\widetilde{\omega})}_{\widetilde{\omega}} \coloneqq {(\offer')}^{b({\omega})}_{{\omega}}$ and $(\widetilde{\offer}')^{s(\widetilde{\omega})}_{\widetilde{\omega}} \coloneqq {(\offer')}^{s({\omega})}_{{\omega}}$ for all ${\omega \in \Omega^{1,2}}$. Because only agent $k$ best responds in $\widetilde{\market}(J^1, J^2)$ and only agents from $J^k$ best respond in $\market$, the offers made by all other agents are unchanged so this statement is clear for these offers. We now consider the offers of agent $k$ and agents $J^k$, respectively.

Fix some $\omega \in \Omega^{1,2}$, and let $j$ be the agent of $\omega$ in $J^k$. By the above, we have $\offer^{-j}_\omega = (\offer')^{-j}_\omega$. Moreover, as agent $j$ is satisfied after the best response sequence, our definition of best responses tells us that $(\offer')^{j}_\omega = (\offer')^{-j}_\omega$ if $\omega \in \Xi$ and $(\offer')^{j}_\omega = (\offer')^{-j}_\omega - \chi^j_\omega$ otherwise. Similarly, as the agent $k$ of market $\widetilde{\market}(J^1, J^2)$ is satisfied after best responding, we have $\widetilde{\offer}^{-k}_{\widetilde{\omega}} = (\widetilde{\offer}')^{-k}_{\widetilde{\omega}}$; and $(\widetilde{\offer}')^{k}_{\widetilde{\omega}} = (\widetilde{\offer}')^{-k}_{\widetilde{\omega}}$ if $\widetilde{\omega} \in \widetilde{\Phi}$ and $(\widetilde{\offer}')^{k}_{\widetilde{\omega}} = (\widetilde{\offer}')^{-k}_{\widetilde{\omega}} - \chi^k_{\widetilde{\omega}}$ otherwise. As $\offer^{-j}_{\omega} = \widetilde{\offer}^{-k}_{\widetilde{\omega}}$ by definition of $\widetilde{\offer}$, the result follows from $\Phi = \Xi \cap \Omega^{1,2}$.
\end{proof}

We are now ready to prove \cref{proposition:convergence-guarantees}.
\begin{proof}[Proof of \cref{proposition:convergence-guarantees}]
Fix $m$-sparse market $\market$ in state $(U^0, \offers^0)$. We prove the theorem by induction on $|I|$. Assume without loss of generality that $|I|$ is a power of two, by add isolated agents if necessary. The base case $|I|=1$ is immediate, so suppose $|I| \geq 2$. Partition $I$ into two non-empty sets $J^1$ and $J^2$ so that $\Omega$ contains at most $m$ trades with endpoints in both $J^1$ and $J^2$. Such a partition exists because the market is $m$-sparse. For notational convenience, also define $J^k \coloneqq J^{k \bmod 2}$ for any $k \geq 1$. We construct the best response sequence $S=S^1 S^2 S^3 \ldots$ of agents from $I$ by concatenating subsequences $S^k$ containing only agents from $J^1$ when $k$ is odd, and only agents from $J^2$ when $k$ is even. After each such subsequence containing agents $J^k$, we guarantee that \textit{all} agents in $J^k$ are satisfied, so our best response sequence will alternate between satisfying the agents in $J^1$ and in $J^2$. The sequence terminates if all agents in $I$ are satisfied simultaneously.

We define the subsequences $S^k$ recursively as follows. Fix ${k \geq 1}$, and let $U^{k-1}$ and $\offers^{k-1}$ be the unsatisfied agents and offers after applying best response sequence $S^1 \ldots S^{k-1}$ to market $\market$ in state $(U^0, \offers^0)$. If $U^{k-1} = \emptyset$, then ${S \coloneqq S^1 \ldots S^{k-1}}$ and we are done. Otherwise, we consider the restricted market $\widehat{\market}(J^{k}, \offers^{k-1})$, which is FS if $\market$ is. By induction hypothesis (as ${|J^k| < |I|}$), this market $\widehat{\market}(J^{k}, \offers^{k-1})$ in state $(\widehat{U}^{k-1}, \widehat{\offers}^{k-1})$ admits a (finite) terminating best response sequence $S^k$. And by \cref{lemma:restricted-market}, all agents $J^k$ are satisfied in market $\market$ after applying this sequence $S^k$ to $\market$ in state $(U^{k-1}, \offers^{k-1})$.

We now prove that $S$ consists of at most $f(m,V)$ subsequences $S^k$. Suppose $S$ has $K \leq \infty$ subsequences. We consider the two-agent market $\widetilde{\market}(J^1, J^2)$ obtained by merging agent sets $J^1$ and $J^2$. \Cref{lemma:merged-market} implies that $\widetilde{\offers}^k$ are the offers after applying $k$ best responses in this two-agent market in state $(\widetilde{U}^0, \widetilde{\offers}^0)$, starting with agent $1$. After every subsequence $S^{k}$ with $k \leq K-1$, at least one offer facing the agents in $J^{k-1}$ has changed, as at least one agent in $J^{k-1}$ is newly unsatisfied. So in the corresponding two-agent market, the agent $k-1 \bmod 2$ is unsatisfied after the $k$th best response, for any $k \leq K-1$. Similarly, both agents are satisfied after the $K$th best response in the two-agent market, so the two-agent market terminates after exactly $K$ best responses (if $K=\infty$, it does not terminate). But by assumption, the two-agent market with at most $m$ trades terminates after at most $f(m,V)$ best responses, so $K \leq f(m,V)$.
\end{proof}

\begin{corollary}
Any $2$-sparse FSM converges almost surely to equilibrium.  Moreover, if \cref{conjecture:two-agent-convergence} holds, then any market with FS agents converges almost surely.
\end{corollary}
\begin{proof}

Suppose two-agent markets converge. This holds for $m$-sparse topologies with $m \leq 2$ by \cref{proposition:two-agents-convergence}, and by assumption for $m \geq 3$ by \cref{conjecture:two-agent-convergence}. Then by \cref{proposition:convergence-guarantees}, any market with FS agents, in any state, admits a finite best response sequence after which the market dynamic terminates. As the dynamic selects the next agent to best respond uniformly at random, it will almost surely select, eventually, a best response sequence that results in an equilibrium.
\end{proof}

\section{Network details}

\subsection{Network Construction}
\subsubsection{BS networks}

We consider a basic network, with $B$ buyer nodes and $S$ seller nodes. Each buyer node $b \in B$ is connected to each seller node $s \in S$ with some probability $r$. Trades can only be made between connected nodes in all networks.

\subsubsection{BIS networks}

We consider a basic network, with $B$ buyers,  $S$ sellers, and $I$ intermediaries. Each buyer $b \in B$ is connected to each intermediary $i \in I$ with some probability $r$, and likewise, each seller $s \in S$ is connected to each intermediary with some probability $r$. Buyers/sellers are not directly connected in this network, and only transact through an intermediary.

\subsubsection{General networks}\label{appendixNetwork}

We consider Erdős-Rényi (ER) networks $G(N, \lambda)$, where $N$ is the number of traders, and $r=\frac{\lambda}{N}$ is the probability of each possible connection being made between the traders (i.e.,$\lambda=N$ means that the graph is fully connected, and $\lambda \approx 1$ means that the graph is sparsely connected and tree-like).

The graph construction process is as follows:

\begin{enumerate}
    \item Create a random undirected ER network $G(N, \lambda)$.
    \item Consider only the largest component of $G(N, \lambda)$, dropping isolated nodes and smaller sub-graphs.
    \item Treat each internal node as an intermediary agent.
    \item Randomly assign leaf nodes to be a buyer or seller (with probability $0.5$ each).
\end{enumerate}

Trades are constructed both ways between two intermediaries (e.g., they are both interested in being the buyer or the seller).

Example networks are displayed in \cref{figERExamples}.

\begin{figure}[H]
\centering
 \begin{subfigure}[b]{0.45\columnwidth}
    \includegraphics[width=\textwidth]{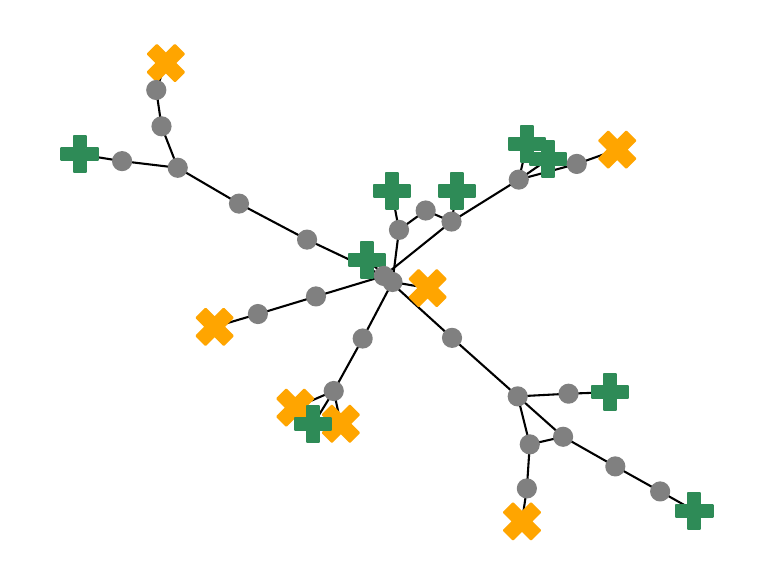}
        \caption{$\lambda=1$}
     \end{subfigure}
     \hfill
     \begin{subfigure}[b]{0.45\columnwidth}
         \centering
         \includegraphics[width=\textwidth]{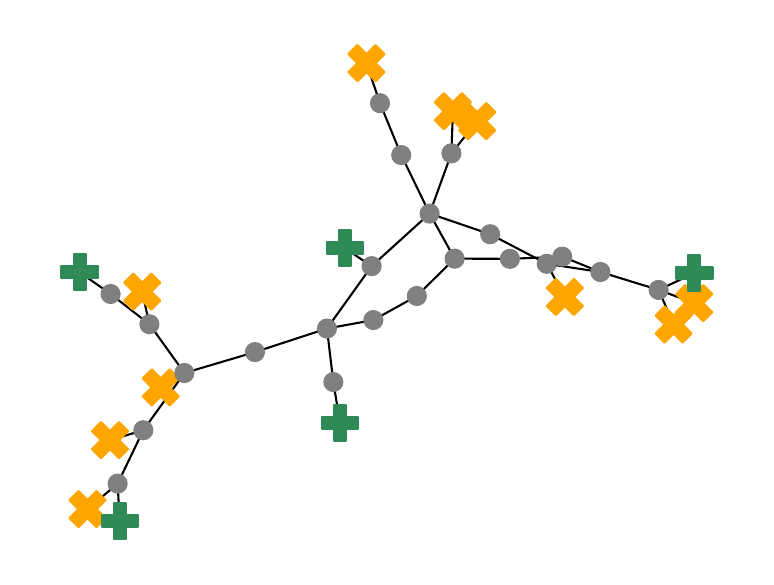}
         \caption{$\lambda=1.5$}
     \end{subfigure}
     \hfill
          \begin{subfigure}[b]{0.45\columnwidth}
         \centering
         \includegraphics[width=\textwidth]{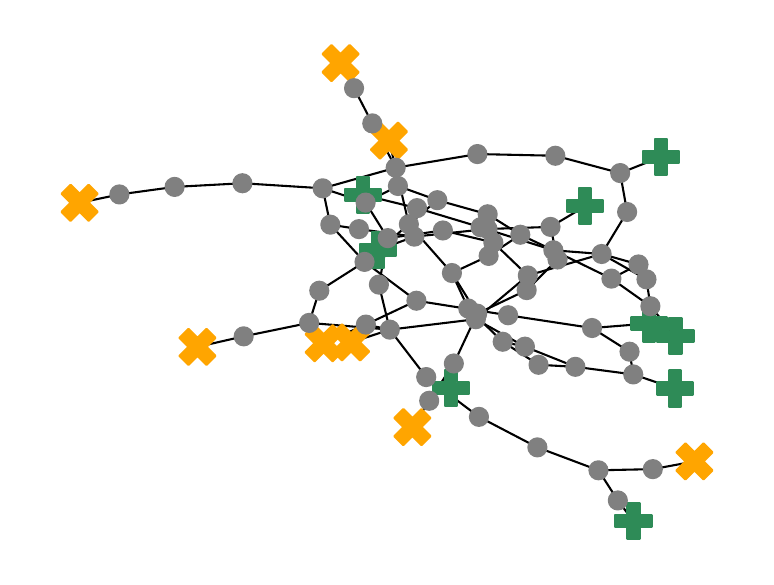}
         \caption{$\lambda=2$}
     \end{subfigure}
      \hfill
               \begin{subfigure}[b]{0.45\columnwidth}
         \centering
         \includegraphics[width=\textwidth]{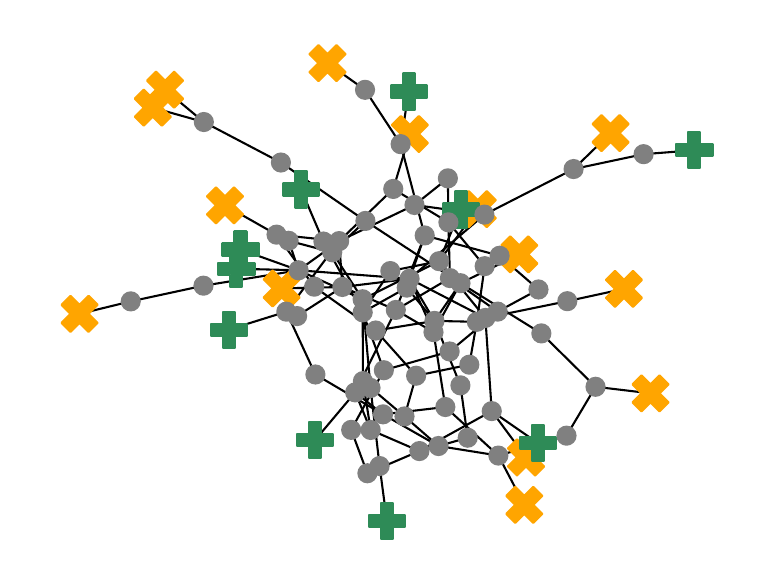}
         \caption{$\lambda=2.5$}
     \end{subfigure}
      \hfill
               \begin{subfigure}[b]{0.18\textwidth}
         \centering
         \includegraphics[width=\textwidth]{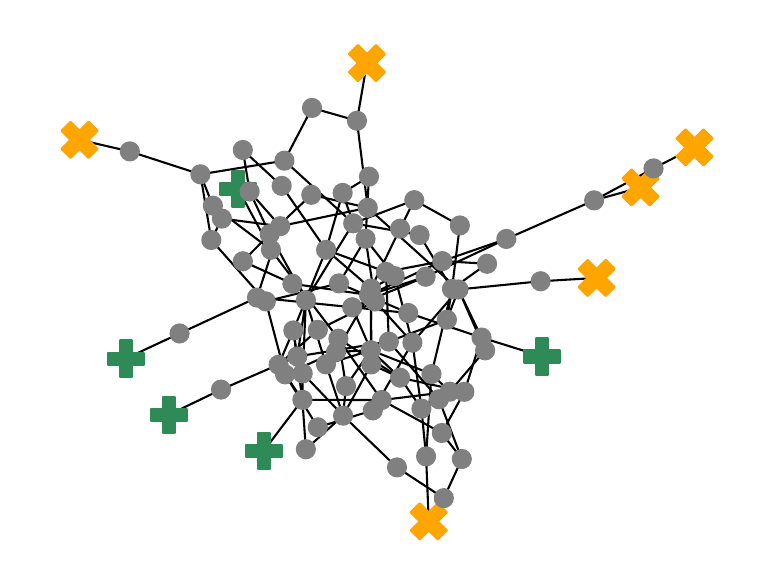}
         \caption{$\lambda=3$}
     \end{subfigure}
     \hfill
    \caption{Example ER networks with $N=100$. Buyers (sellers) are shown as green $+$ (orange $\times$). Intermediaries are gray $\cdot$.}
    \label{figERExamples}
\end{figure}

\section{Additional Plots}

\begin{figure*}[!htb]
\centering
\begin{subfigure}[b]{0.3\textwidth}
\centering
\includegraphics[width=\textwidth]{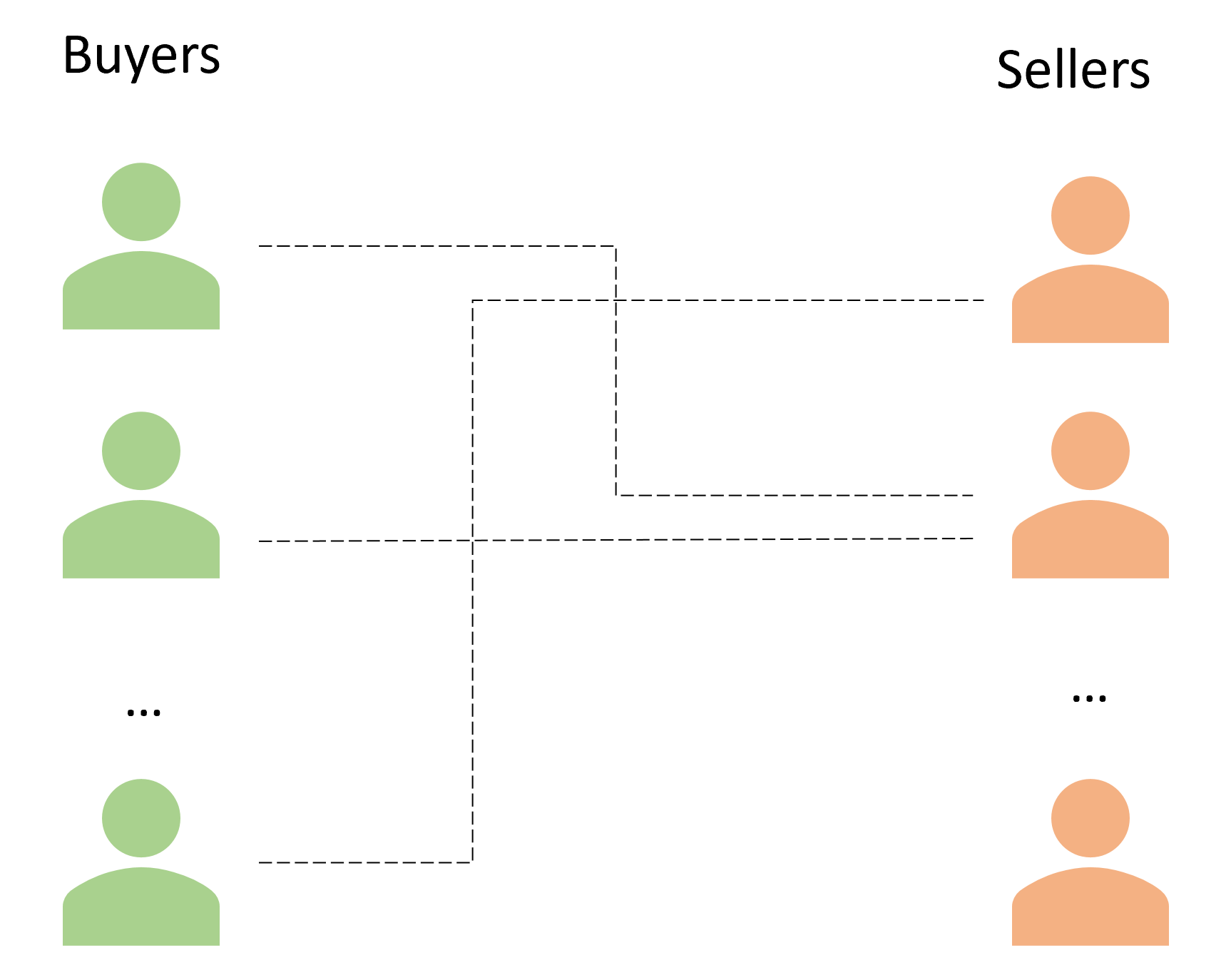}
\caption{BS}\label{figBS}
\end{subfigure}
\hfill
\begin{subfigure}[b]{0.3\textwidth}
\centering
\includegraphics[width=\textwidth]{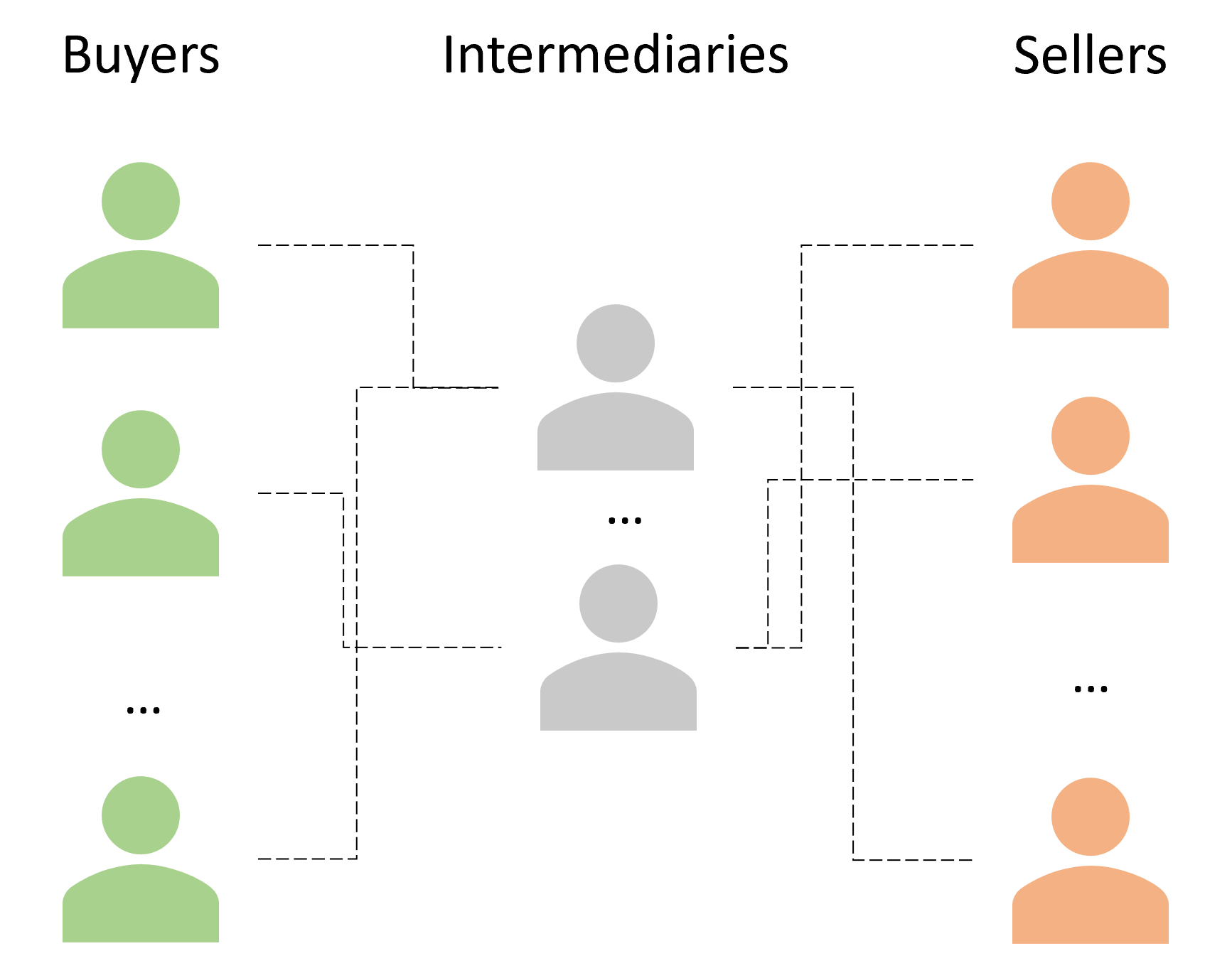}
\caption{BIS}\label{figBIS}
\end{subfigure}
\hfill
\begin{subfigure}[b]{0.3\textwidth}
\centering
\includegraphics[width=\textwidth]{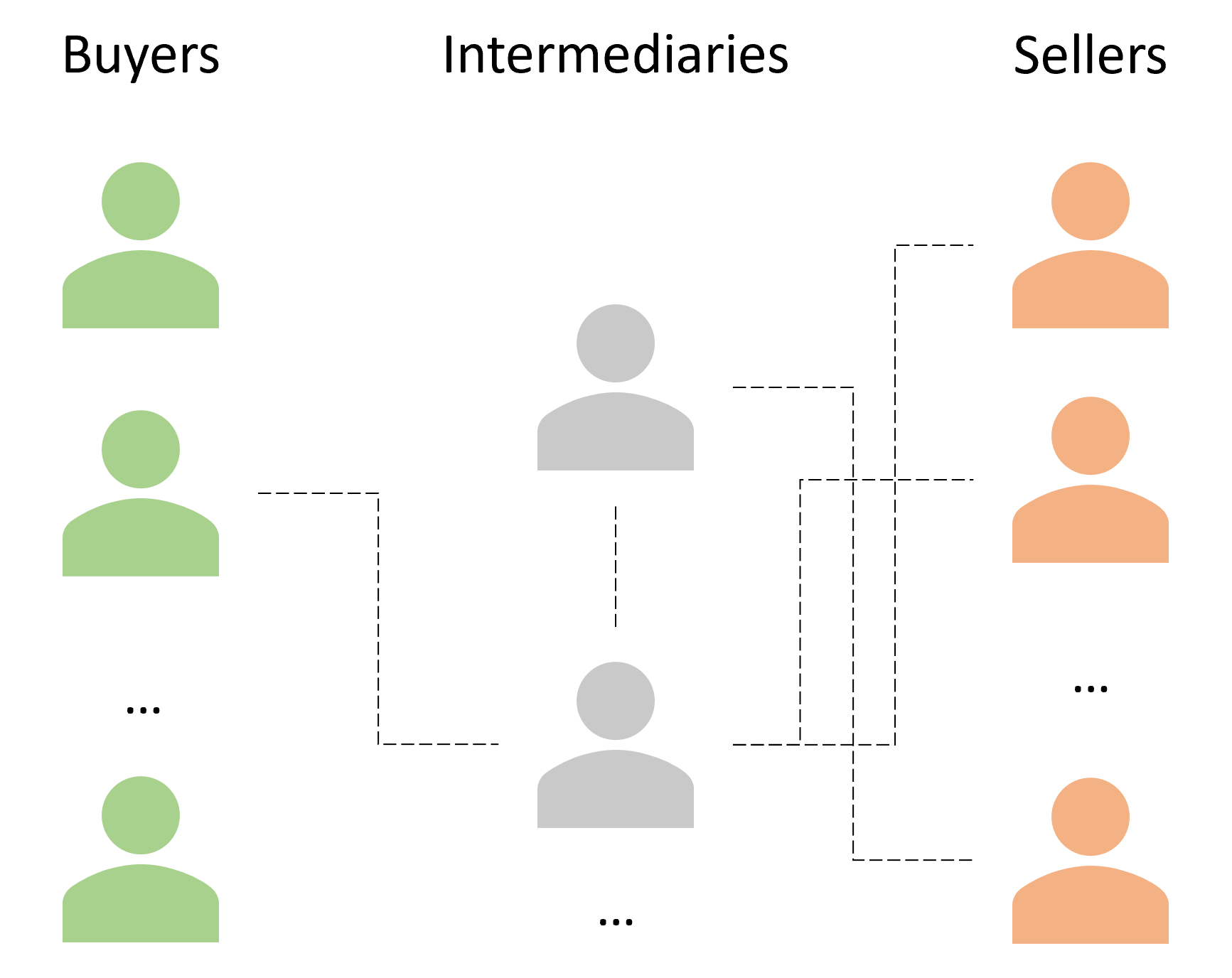}
\caption{General}\label{figGeneral}
\end{subfigure}
\caption{Example network topologies. }
\label{figNetworks}
\end{figure*}

\begin{figure*}[!htb]
    \centering
    \begin{subfigure}[b]{0.32\textwidth}
    \includegraphics[width=\textwidth]{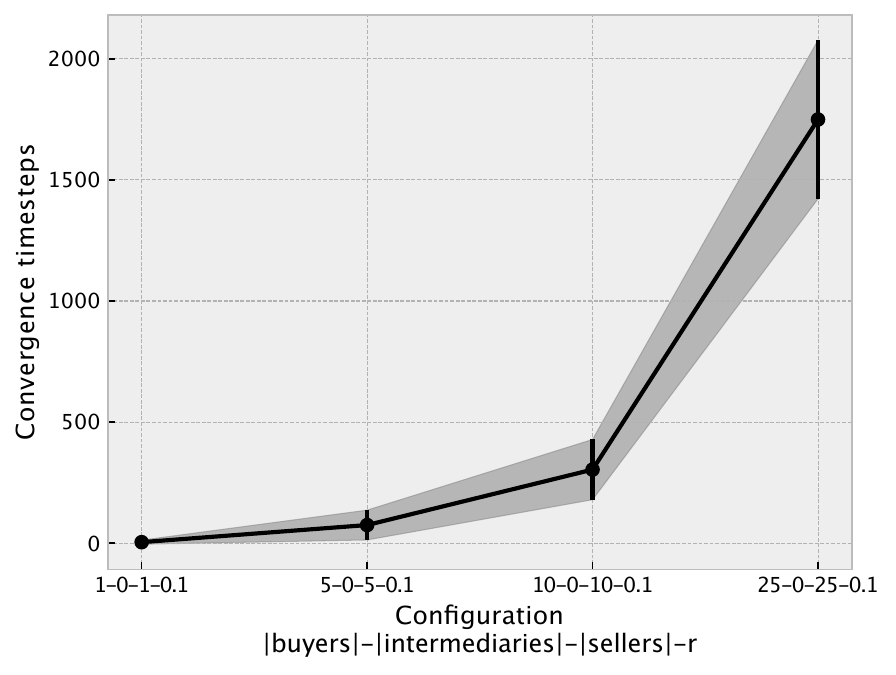}
    \caption{BS}
        \label{figConvergenceRatesBS}
     \end{subfigure}
     \hfill
         \begin{subfigure}[b]{0.32\textwidth}
    \includegraphics[width=\textwidth]{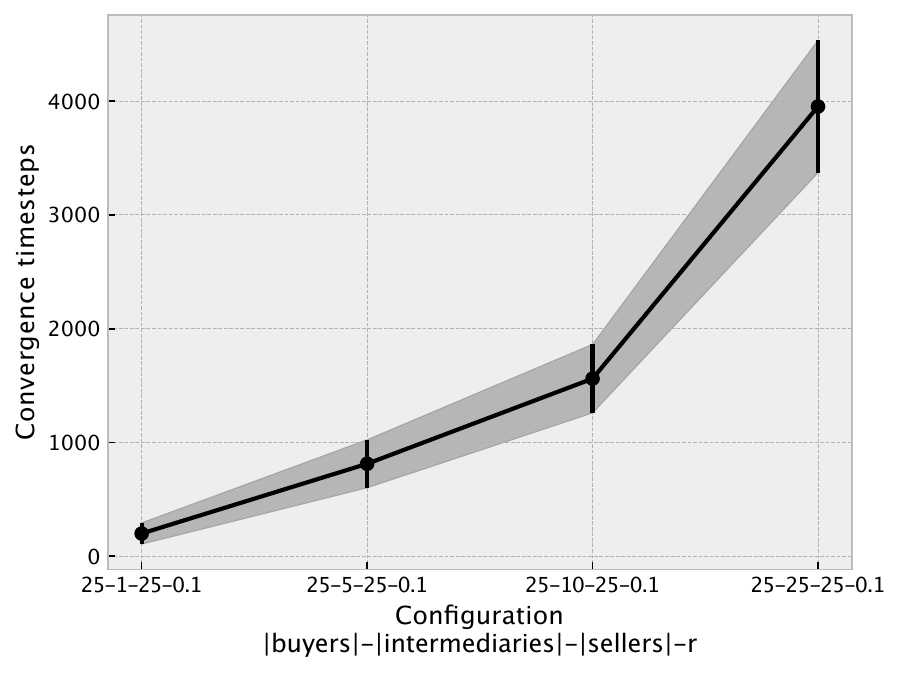}
    \caption{BIS}
        \label{figConvergenceRatesBIS}
     \end{subfigure}
     \hfill
         \begin{subfigure}[b]{0.32\textwidth}
    \includegraphics[width=\textwidth]{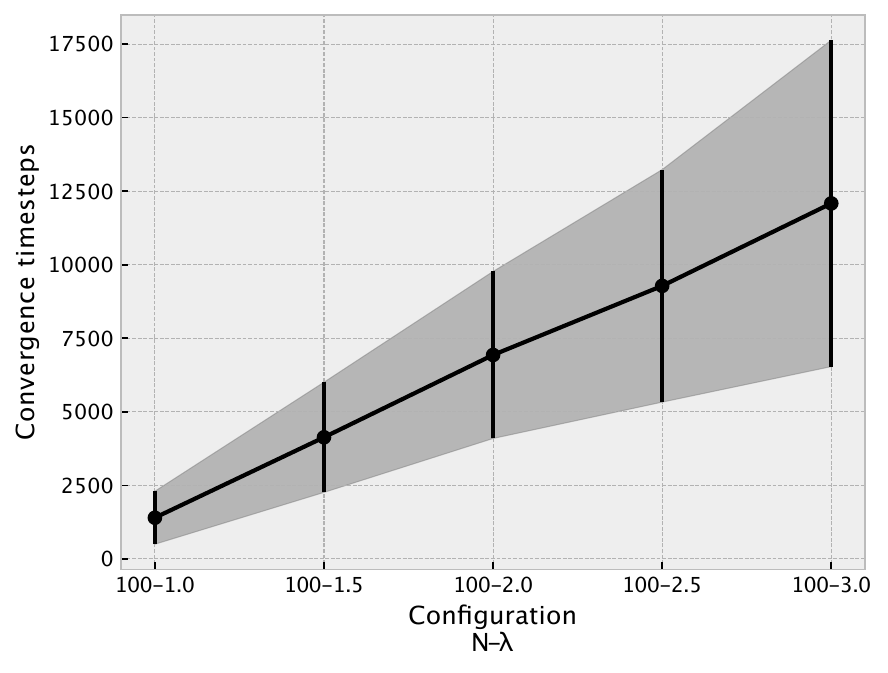}
    \caption{General}
        \label{figConvergenceRatesGeneral}
     \end{subfigure}
     \hfill
    \caption{Convergence rates by market topology. The $x$-axis varies the size of the market, and the y-axis shows the resulting time (number of best response updates) to convergence.}
    \label{figConvergenceRates}
\end{figure*}

\begin{figure*}[!tb]
    \centering
    \begin{subfigure}[b]{0.45\textwidth}
    \includegraphics[width=0.88\textwidth]{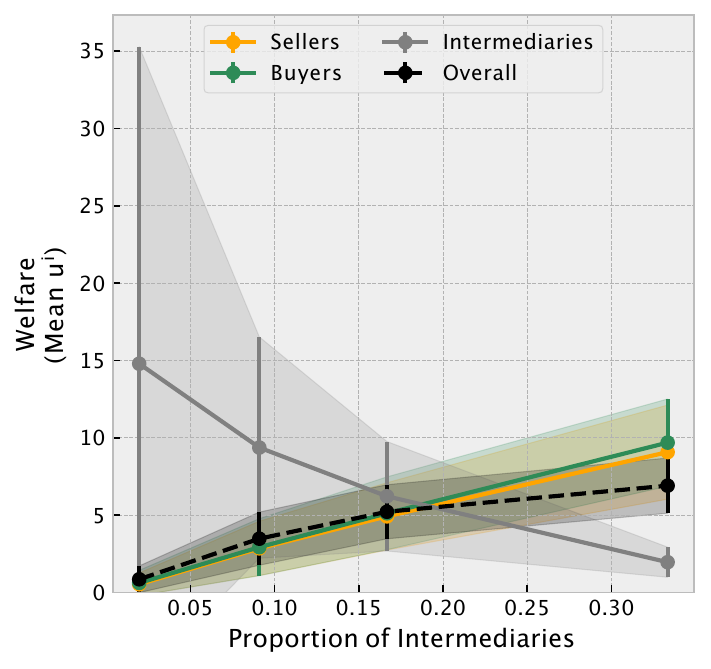}
     \caption{BIS Network. The $x$-axis varies the proportion of intermediaries. As the proportion of intermediaries increases, their welfare decreases, and the welfare of the buyers/sellers increases.}
    \end{subfigure}
     \hfill
    \begin{subfigure}[b]{0.45\textwidth}
    \centering
\includegraphics[width=0.88\textwidth]{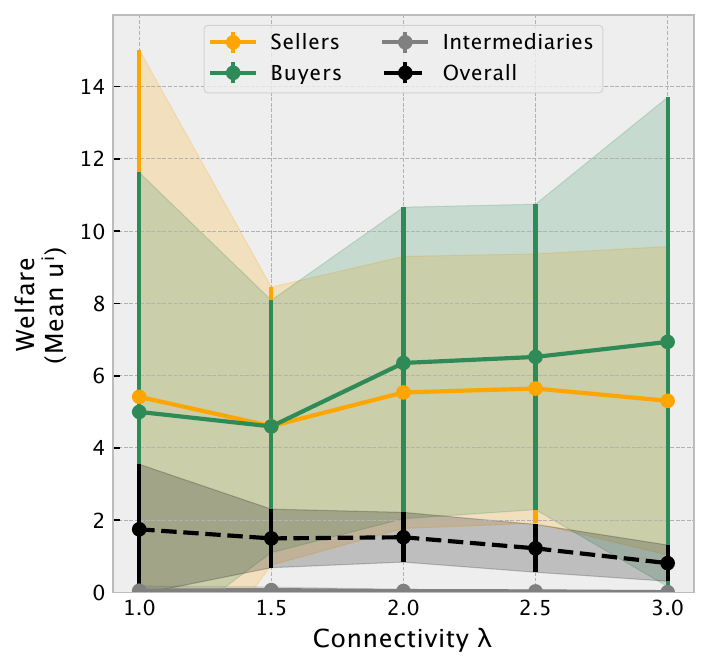}
    \caption{General network. The $x$-axis varies the connectivity, showing the welfare is relatively constant across configurations.}
    \label{figGeneralWelfare}
     \end{subfigure}
      \hfill
    \caption{Additional social welfare (mean $u_i$) plots by agent composition.}
    \label{figAdditionalWelfare}
\end{figure*}

To support the experimental findings, we include additional plots. Example network topologies are visualized in \cref{figNetworks}, and the convergence rates based on market topologies across different sizes is shown in \cref{figConvergenceRates}. Social welfare across BIS and General networks is given in \cref{figAdditionalWelfare}.

\end{document}